\documentclass[sigconf]{acmart}
\usepackage{mathrsfs}
\usepackage{blkarray}
\usepackage{amsmath}
\usepackage{amsthm}
\usepackage{graphicx}

\usepackage{amssymb}
\usepackage{xcolor}
\usepackage{algorithmicx}
\usepackage{bm}
\usepackage{hyperref}
\usepackage{cleveref}
\crefname{section}{§}{§}
\Crefname{section}{§}{§}
\usepackage{fancyhdr,lipsum}
\usepackage{array}
\newcolumntype{C}[1]{>{\centering}p{#1}}
\usepackage{longtable,tabularx,float}
\usepackage{multirow}
\usepackage{subfigure}
\usepackage[most]{tcolorbox}

\usepackage[compact]{titlesec}
\titlespacing*{\subsection}
{0pt}{1.2ex plus .1ex}{1.ex plus .1ex}
\titlespacing*{\subsubsection}
{0pt}{1.0ex plus 0ex}{0.8ex plus 0ex}

\begin{document}

\title{New Wide Locally Recoverable Codes with Unified Locality}

\author{Liangliang Xu}
\affiliation{%
  \institution{Xidian University}
  \city{Xi'an}
  \state{Shanxi}
  \country{China}
}
\email{xuliangliang@xidian.edu.cn}

\author{Mingfeng Tang}
\affiliation{%
  \institution{Xidian University}
  \city{Xi'an}
  \state{Shanxi}
  \country{China}
}
\email{fengming.tang553555@outlook.com}

\author{Tingting Chen}
\affiliation{%
  \institution{Xidian University}
  \city{Xi'an}
  \state{Shanxi}
  \country{China}
}
\email{chentingting@xidian.edu.cn}

\author{Qiliang Li}
\affiliation{%
  \institution{University of Science and  Technology of China}
  \city{Hefei}
  \state{Anhui}
  \country{China}
}
\email{leeql@mail.ustc.edu.cn}

\author{Min Lyu}
\affiliation{%
  \institution{University of Science and  Technology of China}
  \city{Hefei}
  \state{Anhui}
  \country{China}
}
\email{lvmin05@ustc.edu.cn}

\author{Gennian Ge}
\affiliation{%
  \institution{Capital Normal University}
  \city{Beijing}
  \country{China}
}
\email{gnge@zju.edu.cn}



\begin{abstract}
Wide Locally Recoverable Codes (LRCs) have recently been proposed as a solution for achieving high reliability, good performance, and ultra-low storage cost in distributed storage systems. 
However, existing wide LRCs struggle to balance optimal fault tolerance and high availability during frequent system events. By analyzing the existing LRCs, we reveal three limitations in the LRC construction which lay behind the non-optimal overall performance from multiple perspectives, including non-minimum local recovery cost, non cluster-topology-aware data distribution, and non XOR-based local coding. Thanks to the flexible design space offered by the locality property of wide LRCs, we present UniLRC, which unifies locality considerations in code construction. UniLRC achieves the optimal fault tolerance while overcoming the revealed limitations.
We implement UniLRC prototype and conduct comprehensive theoretical and system evaluations, showing significant improvements in reliability and performance over existing wide LRCs deployed in Google and Azure clusters.
\end{abstract}



\keywords{Distributed Storage System, Locally Recoverable Codes, Erasure Coding}


\maketitle

\section{Introduction}
Large-scale distributed storage systems~(DSSs) now store vast data, with the majority protected by erasure codes \cite{wang2023design, zhou2024cord, CEPHEC, HDFS_EC, shvachko2010hadoop, pan2021facebook, muralidhar2014f4}. As data volume continues to grow exponentially, minimizing overhead has become increasingly critical. A recent strategy to reduce storage overhead involves using wide stripes for encoding, which utilizes ultra-low parity redundancy. These types of codes, known as ``wide codes'', have been deployed in commercial DSSs, such as VAST \cite{VAST} and Google \cite{kadekodi2023practical}.

At a high level, for configurable parameters $n$ and $k$ (where $k<n$), erasure codes compose a \textit{stripe} of $n$ blocks, consisting of $k$ original uncoded \textit{data blocks} and $n-k$ coded \textit{parity blocks}.
Here, $n$ is the stripe width, and $k/n$ is the 
storage efficiency, also known as the code rate. Wide codes, designed to maximize the code rate $k/n$, focus on large $n$ and $k$ for ultra-low redundancy.
VAST has reported a ($154,150$) code with a $0.974$ code rate \cite{VAST}, and Google has tested a ($105,96$) code with a $0.914$ code rate \cite{kadekodi2023practical}. The stripe width of the these wide codes is much larger than typical storage systems, such as the $>10\times$ width of ($9,6$) code in traditional DSSs \cite{ovsiannikov2013quantcast, HDFS_EC, ford2010availability}.
In recent years, wide codes have increasingly been designed as Locally Recoverable Codes (LRCs) \cite{hu2021exploiting, kadekodi2023practical}, as LRCs significantly reduce recovery costs by adding only a small additional overhead for local parity storage.

However, the ultra-low redundancy and large stripe width of wide LRCs introduce new challenges in balancing reliability and performance during frequent events such as normal read, degraded read, and single-block recovery (reconstruction). The main challenges are as following:
 (1) Extreme code rates require very few parity blocks, making it difficult to balance fault tolerance and recovery cost. LRCs have two types of parity blocks: more global parity blocks provide higher fault tolerance, while more local parity blocks reduce recovery cost.
(2) Wide stripes are difficult to deploy in DSSs due to the hierarchical cluster topology and asymmetric network bandwidth, with oversubscription ratios typically ranging from $5:1$ to $20:1$ \cite{benson2010network,vahdat2010scale,cisco,ahmad2014shufflewatcher}. The varying block types and quantities in wide LRCs further complicate mapping stripes to the cluster topology. Ensuring fault tolerance at the cluster level while minimizing cross-cluster traffic remains a significant challenge.


Fortunately, the locality property of wide LRCs offers significant design flexibility for reliability and performance. For example, factors like the local group size, data distribution across groups, and the linearity of local parity block offer a versatile design space.

We study the locality of existed wide LRCs from three key perspectives: (1) \textit{Recovery locality}: This measures the average number of blocks accessed during reconstruction, reflecting recovery performance in frequent single-failure events \cite{rashmi2013solution}. Previous LRC designs prioritized distance optimality (maximizing fault tolerance) and high code rates \cite{tamo2014family, sathiamoorthy2013xoring}, but at the expense of optimal recovery locality. 
(2) \textit{Topology locality}:
This utilizes inner-cluster access to minimize cross-cluster traffic. ECWide \cite{hu2021exploiting} is a state-of-the-art data placement strategy that improves topology locality in wide LRCs to reduce cross-cluster recovery traffic. However, as ECWide primarily focuses on placement optimization without addressing fundamental limitations in code construction, it fails to fully optimize cross-cluster recovery traffic and can degrade the performance of other operations, such as normal read (see details in \cref{sec:study_locality}).
(3)\textit{\texttt{XOR} locality:} This uses \texttt{XOR}-based local parity computation to avoid the high multiplication complexity \cite{isa_l}. However, existing distance-optimal code designs often prioritize high reliability, sacrificing \texttt{XOR} locality in local group coding \cite{kadekodi2023practical,tamo2014family,sathiamoorthy2013xoring}.

These challenges motivate us to design a new family of wide LRCs, \textit{UniLRCs}, which provide \textbf{\underline{uni}}fied locality, seamlessly integrating \textit{recovery locality}, \textit{topology locality}, and \textit{\texttt{XOR} locality} in code construction, effectively balancing reliability and performance. UniLRCs are constructed using a generator matrix derived from a Vandermonde matrix, followed by a series of matrix decompositions to tightly couple between local and global parity blocks.
UniLRCs are \textit{distance optimal}, achieving the theoretical Singleton bound on fault tolerance while addressing locality limitations.
Our key contributions include:

\begin{itemize}
    \item We study the multidimensional locality of wide LRCs deployment in practical Google and Azure storage clusters, covering \textit{recovery locality}, \textit{topology locality}, and \textit{\texttt{XOR} locality}. Our findings highlight several limitations in existing wide LRCs implementation, which leads to inefficiencies in common events such as normal read, degraded read, and reconstruction (\cref{sec:study_locality}).
    \item We present UniLRCs, a novel family of wide LRCs designed to balance performance and reliability. UniLRCs are constructed using a generator matrix derived from a Vandermonde matrix, followed by matrix decompositions to tightly couple local and global parity blocks (\cref{sec:construction}). This approach ensures \textit{distance optimality}, while addressing the limitations of \textit{recovery locality}, \textit{topology locality}, and \textit{\texttt{XOR} locality} (\cref{sec:system_unilrc}).
    \item We perform a theoretical comparison of different wide LRCs deployed in Azure and Google DSSs, using multiple metrics to model the overhead of each code. Our results demonstrate that UniLRC outperforms existing wide LRCs across various aspects, including load balancing during normal read, degraded read and reconstruction cost, and mean-time-to-data-loss (MTTDL) reliability (\cref{sec:the_ana}).
    \item We implement the UniLRC prototype and conduct extensive system evaluations. 
    Compared with the state-of-the-art Google wide LRC, UniLRC achieves a $27.46\%$ increase in normal read throughput,  
   a $33.15\%$ reduction in degraded read latency, and a $90.27\%$ increase in recovery throughput (\cref{sec:system_performance}).
\end{itemize}

\section{Background and Motivation}\label{background}
\subsection{Erasure Coding Basics}\label{sec:ec_basic}
Erasure coding is commonly used in many DSSs \cite{huang2012erasure, sathiamoorthy2013xoring,HDFS_EC,CEPHEC,kadekodi2023practical}, because it provides the same fault tolerance as replication but at a significantly lower cost.
Typically described as an $(n,k)$ erasure code, it encodes $k$ data blocks into $n-k$ parity blocks to form a stripe of width $n$. We begin by defining key terms in erasure coding.

\begin{definition}[Linear code \cite{huffman2010fundamentals}]
\label{def:code_matrix}
A linear $(n, k)$ code over $\mathbb{F}_q$ is a $k$-dimensional subspace $C \subseteq \mathbb{F}^n_q$ where $q$ is a prime power, $\mathbb{F}_q$ is a galois field ($GF$)  of $q$ elements, and $n > 0$. 
It is customary to think of $C$ as the image of an encoding map $Enc: \mathbb{F}^k_q \rightarrow \mathbb{F}^n_q$ for some $k \leq n$. This encoding may be expressed in matrix form as,
$$Gx = y,$$
where $G$ is an $n \times k$ matrix called the generator matrix, $x$ is the message corresponding to $k$ data blocks, and $y$ is the codeword corresponding to $n$ blocks of the stripe. 
The fraction $\frac{k}{n}$ is the code rate. Each component of $y$ is called a codeword symbol (block). The minimum distance $d$ of $C$, where $d=\min_{y\neq z,y,z\in C}d(y,z)$, and 
\[d(y,z)=|\{i:y_i\neq z_i\}|,\]
is called the distance between two codewords $y,z$.
\end{definition}

An \((n, k)\) code with a higher \( k/n \) code rate results in lower storage cost. Meanwhile, an \((n, k)\) code with a larger \( d \) (minimum distance) provides higher fault tolerance. 
An $(n, k)$ code with the \textit{Maximum Distance Separable (MDS)} property \cite{MDS}, such as Reed-Solomon codes \cite{reed1960polynomial}, can tolerate up to $n-k$ concurrent failures. However, for recovering a single-block failure, MDS codes commonly require retrieving \( k \) blocks, which incurs significant recovery cost. 
To address this, \textit{locally recoverable codes (LRCs)} are introduced, reducing the number of blocks required for single-block recovery from \( k \) to \( r \) (\( r \ll k \)).

\begin{definition}[$(n, k, r)$-LRC \cite{asteris2013xoring,huang2012erasure}]
    An $(n, k, r)$-LRC $C$ is a linear code of dimension $k$ and code length $n$, with the following property: for any $x=(x_1,\ldots,x_n)\in C$, any block $x_i$ can be recovered by at most $r$ other blocks $\{x_{\ell_1},\ldots,x_{\ell_r}\}$ where $\ell_j\neq i$ for any $j$. The index set $\{i,\ell_1,\ldots,\ell_r\}$ form a local (recovery) group of $C$. The parameter $r$, known as the locality parameter, satisfies $1 \leq r \leq k$.  A local parity block in an LRC is a parity block computed within a local group of blocks, while a global parity block is a parity block computed across all $k$ data blocks. Let $g$ and $l$ denote the number of global and local parity blocks, respectively.
\end{definition}

LRCs have a bound on the minimum distance $d$ as follows.

\begin{theorem}[Singleton bound and distance optimal LRCs\cite{gopalan2012locality,papailiopoulos2012locally,xing2022construction,guruswami2018long}] \label{thm:singleton}
The minimum distance $d$ of LRC$(n, k, r)$  satisfies $$ d \leq n - k - \left\lceil \tfrac{k}{r} \right\rceil + 2.$$
Moreover, any LRC$(n, k, r)$ that meets the equal condition is called the distance optimal LRC.

If $r\geq d-2$ and $(r+1)|n$, the condition for distance optimal can be reformulated as:
$$n-k-\frac{n}{r+1} = d-2.$$
\end{theorem}

If an LRC is distance optimal, it means achieving the maximum possible on fault tolerance with given code parameters.

\subsection{Wide LRCs}
Data redundancy is a major component of DSSs' storage cost. Every percentage reduction in storage cost can translate into millions of dollars in capital, operational, and energy savings, minimizing this cost remains a highly attractive and practical objective \cite{kadekodi2023practical}.
Wide LRCs have recently gained attention due to the demand for extreme storage efficiency \cite{hu2021exploiting,kadekodi2023practical}, focusing on large $n$ and $k$ to maximize the code rate $k/n$.
The extreme code rate and wide stripe of wide LRCs introduce new challenges in terms of reliability and performance, as outlined below.

$\bullet$ \textbf{Extreme code rate with limited parity blocks.} Many industrial DSSs employ high-rate erasure codes, such as VAST's ($154, 150$) code with a 0.974 code rate \cite{VAST} and Google's ($105, 96$) code with a 0.914 rate \cite{kadekodi2023practical}.
In these systems, high-rate erasure codes are typically defined as having a code rate of $0.85$ or higher \cite{VAST, kadekodi2023practical}. The extreme code rate implies a strict limitation on parity redundancy. 
The strict parity constraints of the wide LRC require a small number of global and local parity blocks (small $g+l$). However, to enhance fault tolerance, more global parity blocks (large $g$) are necessary, while reducing single-block recovery traffic requires more local parity blocks (large $l$). Achieving the right balance between global and local parity blocks for optimal fault tolerance and performance is therefore challenging.
    
$\bullet$ \textbf{Wide stripes with cluster-topology-aware deployment.} Wide LRCs are commonly used with a large stripe width, typically defined as $25-504$ \cite{VAST, kadekodi2023practical}, approximately $2.8\times$ to $56.0\times$ that of Google’s traditional cluster deployment of ($9,6$) codes \cite{ford2010availability}. Wide stripes require more disks to store data, and the larger storage scale introduces multi-tiered topologies in DSSs, such as the cluster topology that can span regions, availability zones, or racks \cite{ford2010availability, zhang2024s}. As a result, the cluster topology introduces asymmetric network bandwidth, with inner-cluster and cross-cluster traffic competing, typically characterized by oversubscription ratios ranging from $5:1$ to $20:1$ \cite{benson2010network,vahdat2010scale,cisco,ahmad2014shufflewatcher}.    
Futhermore, wide LRCs involve three types of blocks: data blocks, local parity blocks and global parity blocks, 
with a significant disparity in the number of blocks, i.e, $l+g \ll k$. 
This large variation in block types and quantities makes mapping stripes to the cluster topology highly complex, especially for uniform access I/Os during normal read, degraded read and reconstruction operations.

To address these challenges, the locality of wide LRCs offers a flexible design space for balancing reliability and performance. For instance, reducing the local group size can reduce recovery traffic, and different data distributions across local groups can lead to varying access patterns. Additionally, the linear properties of local parity blocks can balance fault tolerance and decoding complexity. We present a study on locality using state-of-the-art wide LRCs deployed in Google and Azure storage clusters.

\subsection{Locality Study of Existed Wide LRCs} 
\label{sec:study_locality}
We select three representative LRCs that are widely deployed in practical systems: Azure-LRC (abbr. ALRC) \cite{huang2012erasure}, Optimal Cauchy LRC (abbr. OLRC) \cite{kadekodi2023practical}, and Uniform Cauchy LRC (abbr. ULRC) \cite{kadekodi2023practical}. 
ALRC was the first LRC proposed and deployed in Microsoft Azure \cite{huang2012erasure}. OLRC and ULRC represent two state-of-the-art wide LRCs deployed in Google storage clusters \cite{kadekodi2023practical}.
We provide a construction example for these LRCs with $n=42, k=30$ in Figure \ref{fig:baseline_lrc_example}.
We focus on system frequent events, such as normal/degraded reads and single-block recovery (reconstruction) \cite{kadekodi2023practical, rashmi2013solution, kolosov2018fault}. 

\begin{figure}[t!]
    \centering
    \includegraphics[scale=0.42]{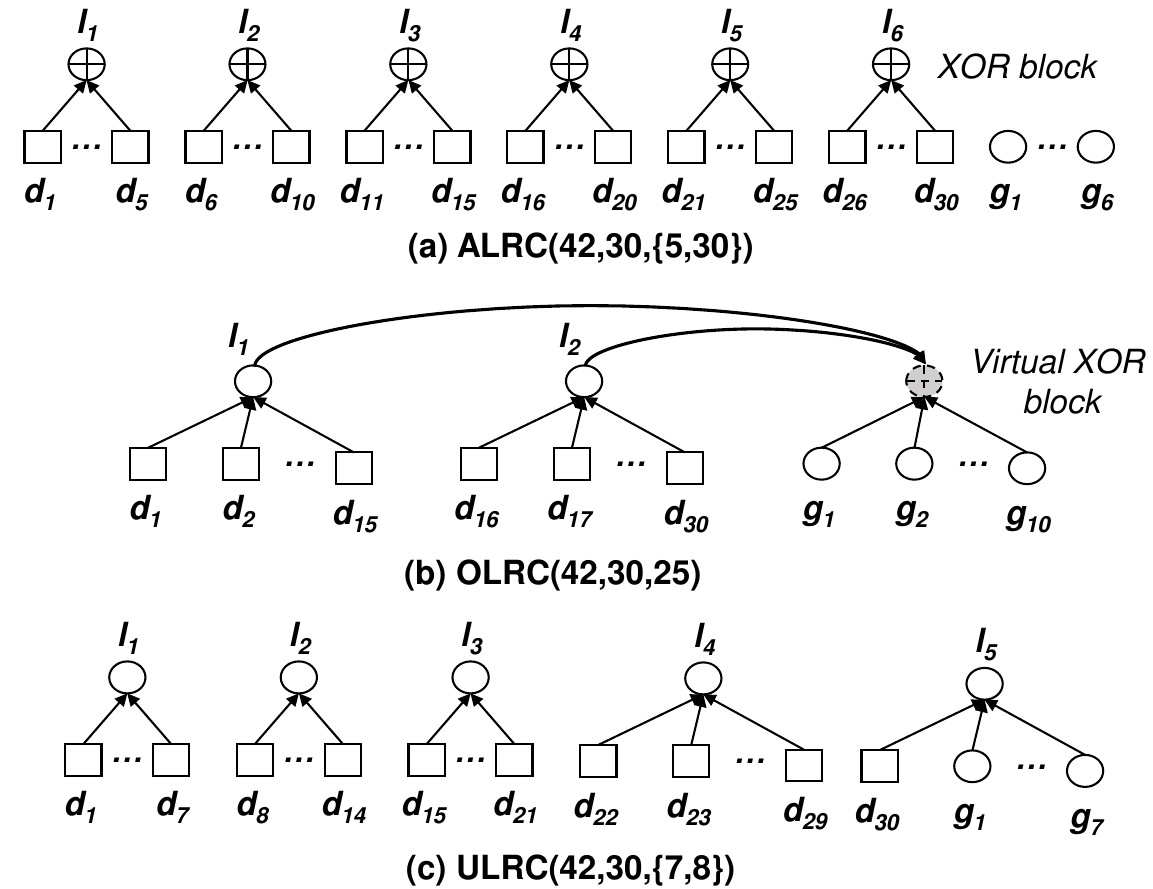}
    \caption{Different wide LRCs with $n=42, k=30$. The LRCs are from Microsoft \cite{huang2012erasure} (ALRC) and Google \cite{kadekodi2023practical} (OLRC and ULRC).
    The terms $d_i$, $l_i$ and $g_i$ mean data block, local and global parity block, respectively.}
    \label{fig:baseline_lrc_example}
\end{figure}

\subsubsection{Recovery Locality.}\label{subsubsection:code_locality}
We define \textit{recovery locality} as \( \bar{r} \), the average number of  blocks required to recover a data or parity block. This \( \bar{r} \) has also been used in previous analyses \cite{kolosov2018fault, kadekodi2023practical}. 
Recovery locality \( \bar{r} \) reflects the average recovery traffic during reconstruction.
Figure \ref{fig:baseline_lrc_example}(a)
shows the ALRC$(42,30,\{5,30\})$, with \textit{recovery locality} \( \bar{r} = \frac{36 \times 5 + 6 \times 30}{42} = 8.57 \). 
Figure \ref{fig:baseline_lrc_example}(b) shows OLRC$(42,30,25)$, with \textit{recovery locality}  \( \bar{r} = 25\), as incurring $25$ blocks for recovering any block no matter data or parity blocks.
Meanwhile, 
Figure \ref{fig:baseline_lrc_example}(c) shows the ULRC$(42,30,\{7,8\})$  with \textit{recovery locality} \( \bar{r} =\frac{24 \times 7 + 18 \times 8}{42} = 7.43 \), which outperforms the other two LRCs.

\textbf{Limitation \#1: Failing to achieve optimal \textit{recovery locality} in existing LRCs.}
The OLRC achieves distance optimality, but its construction condition (i.e., $gl^2 < k+gl$) requires a small number of locality parity blocks (i.e., small $l$), resulting in 
large group size with a large \textit{recovery locality} \( \bar{r} \).
Meanwhile, ULRC offers a good trade-off between reliability and \textit{recovery locality}, though it is not distance optimal (see Theorem \ref{thm:code_distance}). 
However, its \textit{recovery locality} is not optimal (see Theorem \ref{thm:locality} for UniLRC achieving the minimum \textit{recovery locality}), as the design focuses on achieving an approximately even local group size. 
For exampe in Figure \ref{fig:baseline_lrc_example}(c), the ULRC$(42,30,\{7,8\})$ features two types of local group sizes, with sizes of 8 and 9.
Specially, prior LRC constructions  aimed at achieving distance optimality and extreme code rates \cite{tamo2014family,sathiamoorthy2013xoring}. As a result, the existing LRCs fail to achieve optimal \textit{recovery locality}, i.e., the minimal \( \bar{r} \).

\begin{figure}[t!]
    \centering
    \includegraphics[scale=0.38]{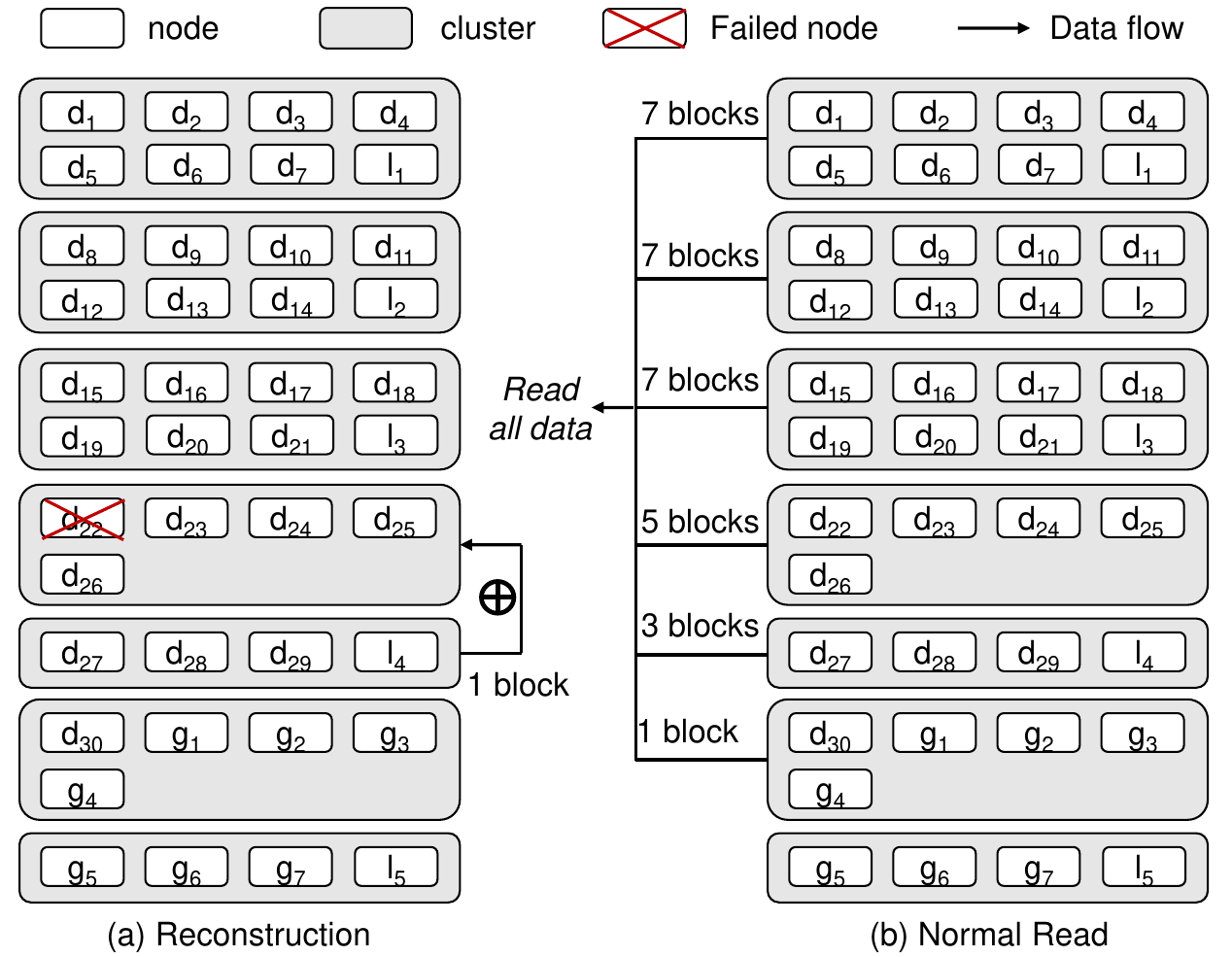}
    \caption{Topology locality of ECWide for ULRC$(42,30,\{7,8\})$.}
    \label{fig:cluster_locality}
\end{figure}

\subsubsection{Topology Locality.}\label{subsubsection:cluster_locality}
The wide LRCs are commonly deployed in a large-scale storage system with multiple clusters \cite{hu2021exploiting}.
The cross-cluster bandwidth is more costly than inner-cluster bandwidth.
Therefore, topology-aware placement can reduce cross-cluster traffic during read and recovery events, which we refer to as \textit{topology locality}.
ECWide \cite{hu2021exploiting} is a state-of-the-art placement strategy that considers \textit{topology locality} of wide LRCs, aiming to minimize cross-cluster recovery cost. 
The core idea of ECWide is to place blocks into a minimum number of clusters while tolerating one-cluster failures. 
This approach reduces the cross-cluster traffic during single-block recovery.
Figure \ref{fig:cluster_locality}(a) shows the application of ECWide on the ULRC$(42,30,\{7,8\})$.
We observe that the first three local groups are placed into exactly three clusters,  while each of the last two local groups are placed into two clusters.
As a result, under reconstructions, 57.1\% blocks incur no cross-cluster recovery traffic, while the remaining blocks incur cross-cluster recovery traffic involving only one block.

\textbf{Limitation \#2: Existing \textit{topology locality} placement optimizing recovery but undermining normal read performance.}
Although ECWide aims to minimize cross-cluster recovery cost, the inherent code construction of ULRC leads to non-optimal cross-cluster recovery traffic.
As shown in Figure \ref{fig:cluster_locality}(a), the recovery of any failed block in the last two local groups still incurs cross-cluster recovery traffic.
Meanwhile, ECWide's grouping of blocks causes load imbalance during normal reads.
Figure \ref{fig:cluster_locality}(b) shows that its data placement results in a $7 \times$ load imbalance across clusters, leading to bottlenecks in the first three clusters (each reading 7 blocks).
As a result, while ECWide improves the recovery performance, it undermines normal read efficiency.

\subsubsection{\texttt{XOR} Locality.} \label{subsubsection:computing_locality}
\texttt{XOR} demonstrates superior coding speed over multiplicity on modern CPUs.
We refer to computing the local parity block using only \texttt{XOR} operations as \textit{\texttt{XOR} locality}.
The \texttt{XOR} local parity block is computed by setting all coefficients to 1 in the corresponding row of the generator matrix. 
For example, the local parity block \( l_1 \) of the ALRC$(42,30,\{5,30\})$ in Figure \ref{fig:baseline_lrc_example}(a) is computed as \( l_1 = \sum_{i=1}^{5} d_i \), simplifying decoding computation in degraded read and reconstruction (e.g., \( d_1 = l_1 + d_2 + \cdots + d_5 \), with only \texttt{XOR} operations).
In contrast, the global parity block \( g_1 = \sum_{i=1}^{30} \alpha_i d_i \), where \( \alpha_i \) represents the encoding coefficients. Therefore, recovering \( g_1 \) requires additional multiplicity (abbr. \texttt{MUL}).
Figure \ref{fig:xor_mul}(a) illustrates the coding throughput on three Intel/AMD CPU families, using the Intel ISA-L coding APIs \cite{isa_l}. The block size is fixed at 64MB, and two blocks are processed using either \texttt{XOR} or \texttt{MUL}+\texttt{XOR}, where \texttt{MUL} $+$ \texttt{XOR} first finds the $GF$ multiplication table for \texttt{MUL} before performing \texttt{XOR}.
The results show that the \texttt{XOR} throughput is consistently higher than \texttt{MUL}+\texttt{XOR}, ranging from 61.33\% to 129.44\%. Additionally, the difference in throughput increases with higher CPU clock speeds. For example, throughput increases from 61.33\% at 2 GHz to 129.44\% at 2.8 GHz, as the time spent finding the $GF$ multiplication table becomes a greater proportion of the total execution time at lower CPU clock speeds.

\begin{figure}[!t]
\centering
\subfigure[Coding Speed]{
\centering
  \includegraphics[width=0.5\linewidth]{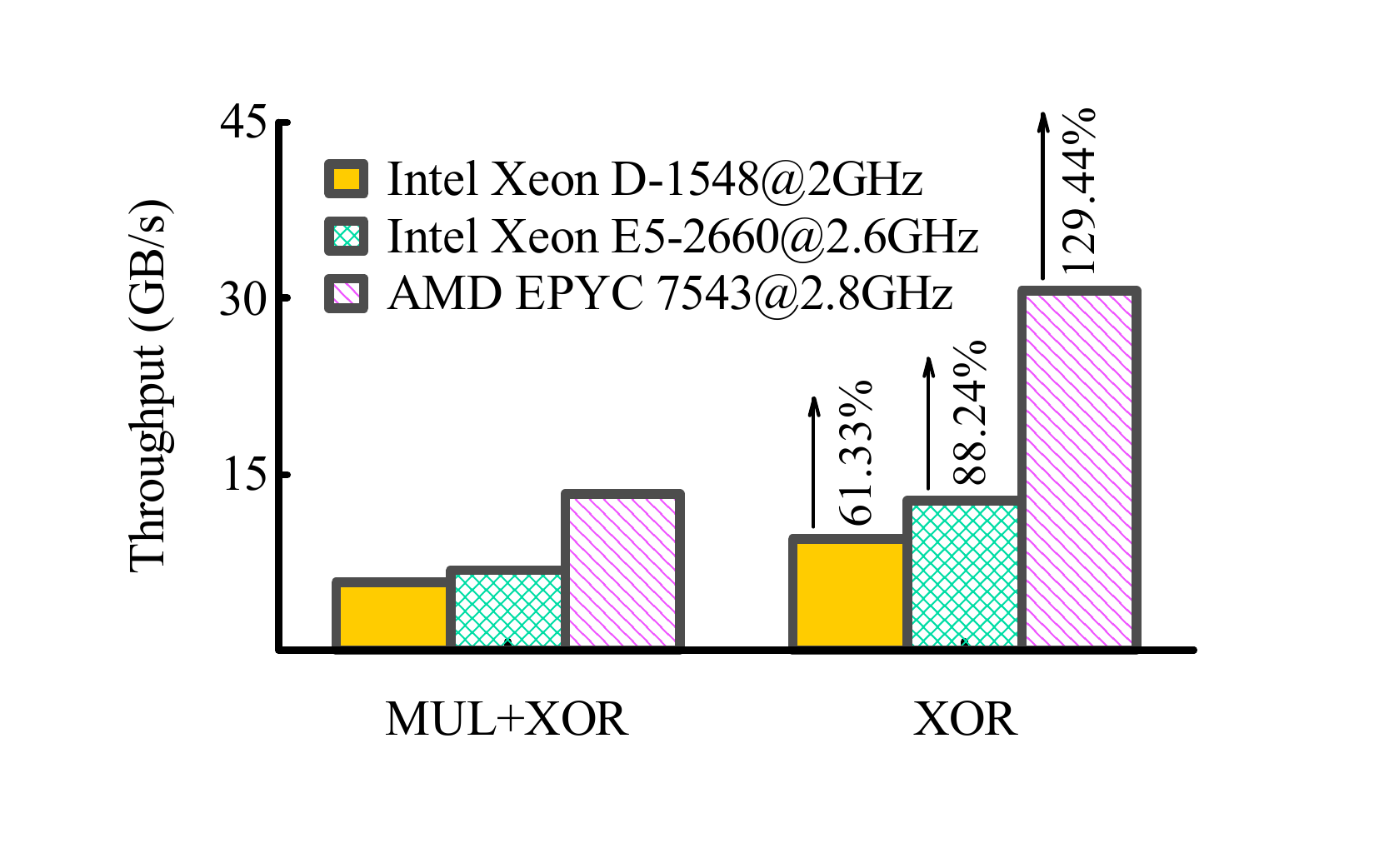}
}
\subfigure[Cost of Single-block Decoding]{
\centering
  \includegraphics[width=0.42\linewidth]{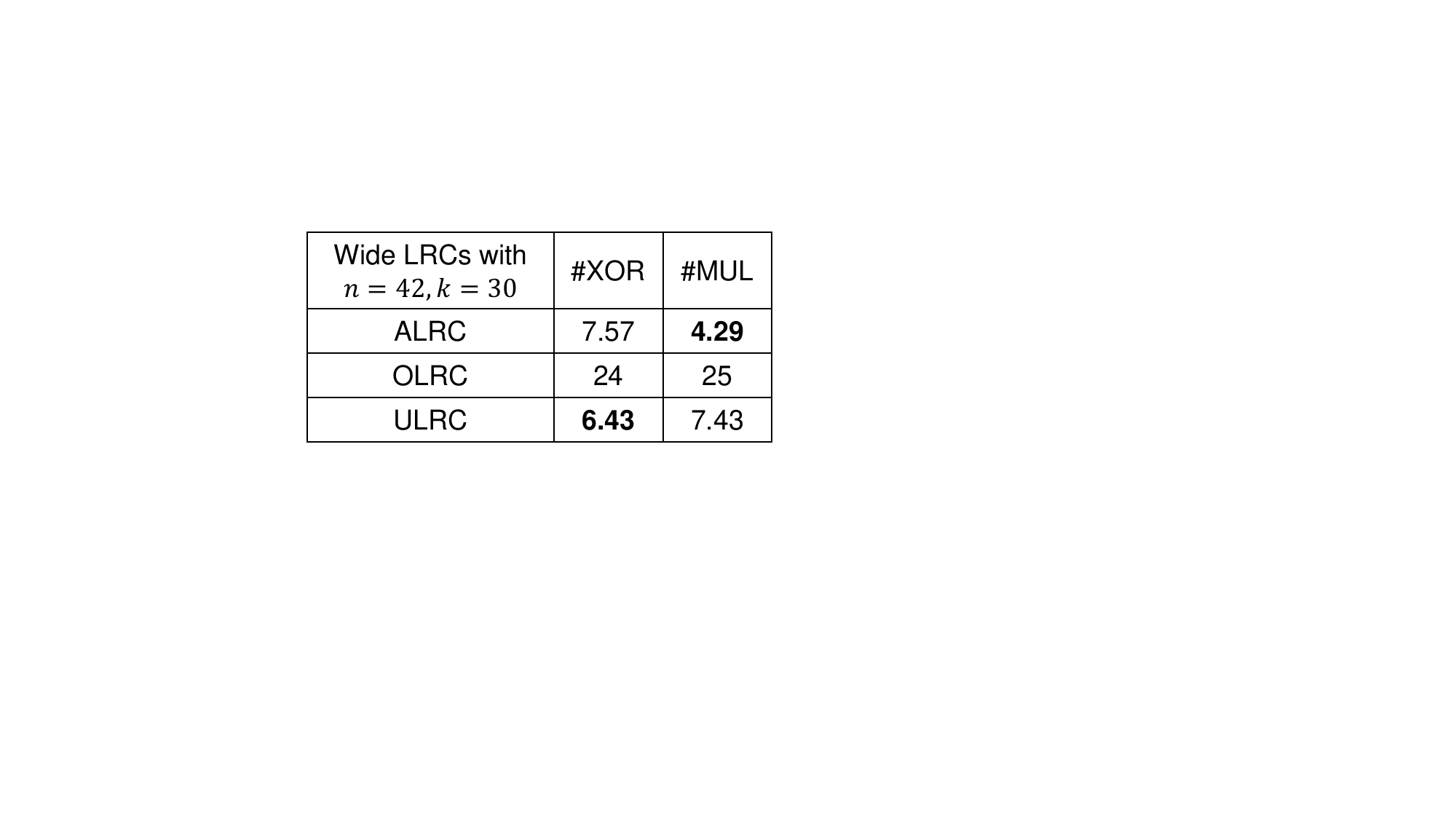}
}
\caption{Comparisons of \texttt{XOR} and \texttt{MUL} under coding computing. (a) Coding throughput on \texttt{XOR} and \texttt{MUL}.
(b) The average times of \texttt{XOR} and \texttt{MUL} for decoding a block with baseline LRCs.
  \label{fig:xor_mul}
}
\end{figure}

\textbf{Limitation \#3: Existing wide LRCs without achieving \textit{\texttt{XOR} locality} yet distance optimal).}
Figure \ref{fig:xor_mul}(b) shows the average times for decoding a failed block using \texttt{XOR} and \texttt{MUL} under baseline LRCs with the same parameters $n=42, k=30$. We observe that all baseline LRCs incur higher decoding times with \texttt{MUL} in single-block decoding. The underlying reason is that the multiple all-1's rows in the generator matrix impose more limitations on fault tolerance, i.e., the code distance \( d \).
However, existing distance-optimal code designs often prioritize code distance, sacrificing \textit{\texttt{XOR} locality} in local group coding \cite{kadekodi2023practical,tamo2014family,sathiamoorthy2013xoring}.
As a result,  it is desirable to have a coupled design on local and global parity blocks, 
for both \textit{\texttt{XOR} locality} and distance optimality.

\begin{table}[!t]
  \centering
  \caption{A detailed comparison of UniLRC with other practical LRCs from Microsoft \cite{huang2012erasure} (ALRC) and Google \cite{kadekodi2023practical}(OLRC and ULRC). Notation: "+", "-" and "$\pm$"  mean "best", "worst" and "in between".}
\begin{tabular}{p{2.7cm}<{\centering}|p{0.9cm}<{\centering}p{0.9cm}<{\centering}p{0.9cm}<{\centering}p{1cm}<{\centering}}
  \hline
 Properties & ALRC &  OLRC & ULRC & \textbf{UniLRC}  \\
  \hline
  Recovery locality & $\pm$ & - & $\pm$ & \textbf{+}  \\

  Topology locality & $\pm$ & - & $\pm$ & \textbf{+}  \\

  \texttt{XOR} locality & $\pm$  & - & - & \textbf{+}  \\

  Distance optimal  & -  & + & - & \textbf{+} \\
  \hline
\end{tabular}

\label{tab:comparisons_existed_solutions}
\end{table}

The main objective of UniLRC is to provide a \textit{unified locality} encompassing  \textit{recovery locality}, \textit{topology locality} and \textit{\texttt{XOR} locality}. 
The core idea of UniLRC is to integrate locality limitations directly into the construction of the generator matrix.
As shown in Table \ref{tab:comparisons_existed_solutions}, UniLRC outperforms the state-of-the-art LRCs in terms of \textit{recovery locality}, \textit{topology locality} and \textit{\texttt{XOR} locality}, while also achieving distance optimality.

\section{Construction of UniLRC} \label{sec:construction}
In this section, we give an overview of UniLRC in \cref{sub:overview}, followed by its construction steps in \cref{subsection:construction} and optimality analysis in \cref{subsec:optimal_analysis}.

\subsection{Overview}\label{sub:overview}
\begin{figure}[h!]
    \centering
    \includegraphics[scale=0.38]{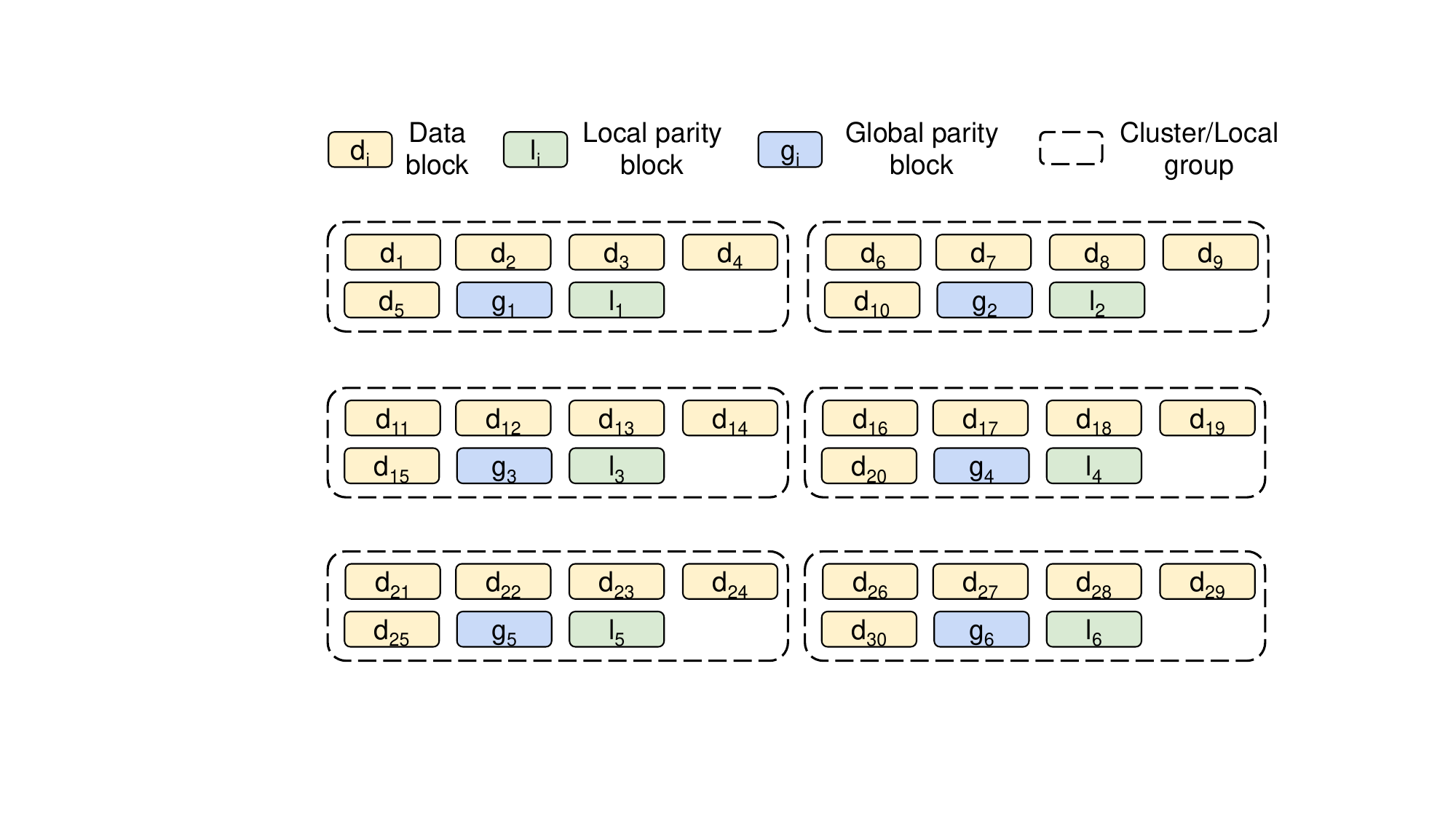}
    \caption{An example of the UniLRC$(42,30,6)$. A local group is assigned to a single cluster, with all data blocks, as well as local and global parity blocks, uniformly distributed across clusters. Local parity blocks are generated using \texttt{XOR}.
    This configuration tolerates up to any $g+1 = 7$ block failures and one cluster failure.}
    \label{fig:unilrc_overview_ex}
\end{figure}

Figure \ref{fig:unilrc_overview_ex} illustrates an example of the UniLRC($n=42, k=30, r=6$) deployed on a DSS with $6$ clusters, 
where each local group mapped into one cluster. 
The data blocks, local parity blocks and global parity blocks are uniformly distributed across $6$ local groups, with each group containing $5$ data blocks, 1 local parity block and 1 global parity block.
The number of global parity blocks is a multiple of the number of clusters.
Global parity blocks are generated using all data blocks, e.g, $g_1 = \sum_1^{30}\alpha_i d_i$, where \( \alpha_i \) represents the encoding coefficients.
Each local parity block is located within a single cluster and generated by \texttt{XOR} operations on the
global parity block and data blocks located within the local group, e.g, $l_1 = \texttt{XOR}\{d_1,\dots,d_5,g_1\}$.

UniLRC achieves the minimum \textit{recovery locality} with $\bar{r} = 6$, addressing   Limitation \#1. 
It maximizes data parallelism during normal read by uniformly accessing $5$ data blocks from any cluster.
Additional, UniLRC ensures zero cross-cluster traffic during single-block recovery.
For example, recovering $d_1$ involves the operation $ \texttt{XOR}\{d_2, d_3, d_4, d_5, l_1, g_1\}$, all of which are co-located in a single cluster.
Thus, UniLRC resolves Limitation \#2.
Furthermore, UniLRC tolerates up to any $g+1 = 7$ block failures and one cluster failure, satisfying the Singleton bound  (refer to  Theorem \ref{thm:singleton}), i.e., achieving the maximum possible fault tolerance. Thus, UniLRC resolves  Limitation \#3.


\subsection{Explicit Construction} 
\label{subsection:construction}
\textbf{Core idea.}
The key construction idea of UniLRC is to integrate \textit{recovery locality}, \textit{topology locality}, and \textit{\texttt{XOR} locality} into the generator matrix, which is based on a series of matrix transformations on the Vandermonde matrix. 
This is achieved by introducing an additional row from the generator matrix of MDS codes, inspired by the construction in \cite{kadekodi2023practical}.
The additional row is designed for local parity blocks, starting with a row of $1$'s, which is then carefully designed through splitting and coupling with the coefficients of the data and global parity blocks.


Recall that the generator matrix $G$ in Definition \ref{def:code_matrix} satisfies $G x = y$, where $x$ is the $k$ data blocks and $y$ is the $n$ blocks of the stripe. We have: 
\begin{equation*}  
G = \left [\begin{array} {c} I_{k} \\ \hline A \end{array}
\right] =
\left [\begin{array} {c} I_{k} \\ \hline \mathcal{G} \\ \hline \mathcal{L} \end{array}
\right],
\end{equation*} 
where $I_{k}$ is a $k\times k$ identity matrix, $A$ is an $(n-k)\times k$ matrix, and $\mathcal{G}, \mathcal{L}$ are submatrices derived from splitting $A$.
Given two positive integer coefficients $\alpha$ (the scale coefficient) and $z$ (the number of clusters), 
we start with a Vandermonde matrix $O$ as the initial generator matrix of order $(\alpha z + 1) \times k$, where $k=\alpha z(z-1)$ :  
$$\begin{bmatrix}
    1 & 1 & \ldots & 1 \\
    g_1 & g_2 & \ldots & g_k \\
    \vdots & \vdots & \ddots & \vdots \\
    g_1^{\alpha z-1} & g_2^{\alpha z-1} & \cdots & g_k^{\alpha z-1} \\
    g_1^{\alpha z} & g_2^{\alpha z} & \ldots & g_k^{\alpha z}
\end{bmatrix}.
$$
The construction process consists of four steps as follows.

\noindent\textbf{Step 1: Splitting the Vandermonde matrix into two parts.}  First, split $O$ into an $\alpha z \times k$ Vandermonde submatrix $\mathcal{G}$ and row vector $l$:
$$\mathcal{G} = \begin{bmatrix}
    g_1 & g_2 & \ldots & g_k \\
    \vdots & \vdots & \ddots & \vdots \\
    g_1^{\alpha z} & g_2^{\alpha z} & \ldots & g_k^{\alpha z} \\
\end{bmatrix}, l = (1, 1, \ldots, 1).$$ 
The submatrix $\mathcal{G}$ is used for generating global parity blocks.

\noindent\textbf{Step 2: Splitting  red all 1 vector into multiple groups.} 
Next, split the $k$ 1’s in the vector $l$ into $z$ equal groups, that is, $\bm{l = (l_{1}, l_{2}, \ldots, l_{z})}$, where each $l_i$ is a length-$\frac{k}{z}$ all 1 vector $(1, 1, \ldots, 1)$. Then we construct a $z \times k$ matrix $L$ as follows, where  $\textbf{0}$ is a length-$\frac{k}{z}$ zero vector.
$$
L = 
\begin{bmatrix}
     \bm{l_{1}} & \textbf{0} & \ldots & \textbf{0} \\
    \textbf{0} & \bm{l_{2}} & \ldots & \textbf{0} \\
    \vdots & \vdots & \ddots & \vdots \\
    \textbf{0} & \textbf{0} & \ldots & \bm{l_{z}} \\
\end{bmatrix}.
$$

\noindent\textbf{Step 3: Combing the global parity blocks within each local group.}
Next, merge the $\alpha z \times k$ matrix $\mathcal{G}$ into a $z \times k $ matrix $\mathcal{G^*}$ by adding every $\alpha$ rows of  $\mathcal{G}$:
$$\mathcal{G^*} = \begin{bmatrix}
    s_{1,1} & s_{1,2} & \ldots & s_{1,k} \\
   \vdots & \vdots & \ddots & \vdots \\
    s_{z,1} & s_{z,2} & \ldots & s_{z,k} \\
\end{bmatrix},$$ 
where $s_{i,j}=\sum_{\gamma=1}^{\alpha}g_j^{(i-1)\alpha+\gamma}$ for $i\in [z], j\in [k]$, i.e., the $i$th row of $\mathcal{G^*}$ is the row sum of the rows of $\mathcal{G}$ from the $((i-1)\alpha+1)_{th}$ row to the $(i\alpha)_{th}$ row. In this way, the $\alpha$ global parity blocks in each local group is combined together.

\noindent\textbf{Step 4: Coupling global and local parity blocks.}
Finally, obtain the coding matrix $\mathcal{L}$ for local parity blocks which are generated by coupling the data  and global parity blocks within each local group, i.e. $\mathcal{L} = 
\mathcal{G^*} +
L$.


Based on the above construction, the complete generator matrix $G$ for UniLRC is as follows: 
\begin{equation*}
G =
\left [\begin{array} {c} I_{k} \\ \hline \mathcal{G} \\ \hline \mathcal{L} \end{array}
\right] =
\left [
\begin{array} {ccccc}
    1 & 0 & 0 & \ldots & 0 \\
    0 & 1 & 0 & \ldots & 0 \\
    0 & 0 & 1 & \ldots & 1 \\
    \vdots & \vdots & \vdots & \ddots & \vdots \\
    0 & 0 & 0 & \ldots & 1 \\
    \hline
    g_1 & g_2 & g_3 & \ldots & g_k \\
    \vdots & \vdots & \vdots & \ddots & \vdots \\
    g_1^{\alpha z} & g_2^{\alpha z} & g_3^{\alpha z} & \ldots & g_k^{\alpha z} \\
    \hline
     &  & \bm{\mathcal{L}_1} & & \\
     &  & \vdots & & \\
   &  & \bm{\mathcal{L}_z} & &  \\
\end{array}
\right ],
\label{equ:construction_unilrc}
\end{equation*}
where $\bm{\mathcal{L}_i}=\mathcal{G}^*_i+(\underbrace{0,\ldots,0}_{(i-1)k/z},\bm{l_i},0,\ldots,0)$  is the $i$th row of $\mathcal{L}$ and $\mathcal{G}^*_i$ is the $i$th row of $\mathcal{G}^*$, for any $i\in [z]$.


\textbf{UniLRC code parameters.} 
Given the number of clusters $z$  and the scale coefficient $\alpha $, we set $g=\alpha z$ to ensure
uniform distribution of the global parity blocks across groups. To tolerate a single cluster failure, we set $r=g$ (refer to Theorem \ref{thm:code_distance} for fault tolerance proof).
Therefore, each local group consists of $1$ local parity block, $\alpha$ global parity blocks and $\alpha z-\alpha$ data blocks.
In summary, UniLRC can be constructed for any parameter sets of the form: 
$(n = \alpha z^2+z, k = \alpha z^2-\alpha z, r=\alpha z)$ with positive integers $\alpha, z$.
The coding of UniLRC is defined over \textit{GF($\mathit{2^8}$)}, making it easy to implement at byte granularity.

\subsection{Optimality Analysis} \label{subsec:optimal_analysis}
We demonstrate that UniLRC exhibits optimality in terms of cost, fault tolerance, and performance.

\textbf{Code rate feasibility.}
The following theorem shows the code rate of UniLRC. 

\begin{theorem} \label{thm:code_rate_bound}
The rate of a UniLRC $(n = \alpha z^2+z, k = \alpha z^2-\alpha z, r=\alpha z)$ code satisfies
\[
\frac{k}{n} = \frac{r}{r+1} (1-\frac{1}{z}) = 1-\frac{\alpha+1}{\alpha z+1}.
\]
\end{theorem}

From the above theorem, the code rate is positively correlated with both the number of clusters $z$, and the scale coefficient $\alpha$.
Furthermore, if $r$ is large (which is common with large stripe width $n$) and $\frac{r}{r+1}$ approaches $1$, the code rate of Theorem \ref{thm:code_rate_bound} approaches $1-\frac{1}{z}$.
Therefore, in a DSS with $z$ clusters, the code rate of UniLRC is limited by the cluster scale with sufficient stripe width.

\begin{figure}[t!]
    \centering
    \includegraphics[scale=0.26]{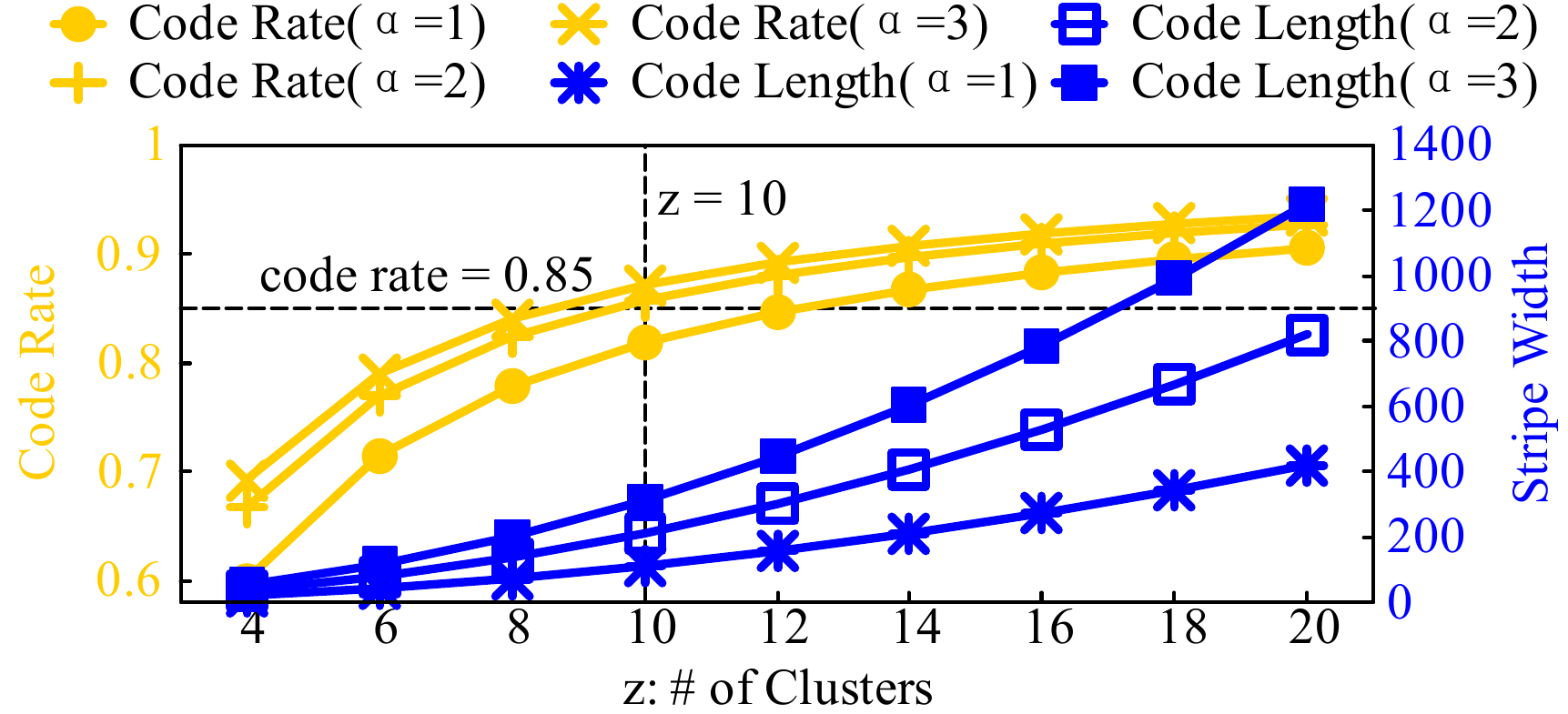}
    \caption{Trade-off on cluster number, scale coefficient, code rate and stripe width for UniLRC with $z \leq 20$ and $\alpha = 1,2,3$.}
    \label{fig:unilrc_codeinfo}
\end{figure}

\textbf{Practical considerations on code rate and stripe width.} 
In known wide-code settings in the industry, the desirable code rate is $\geq 0.85$, while the stripe width typically ranges from $25$ to $504$ \cite{kadekodi2023practical, VAST}.
We show the trade-off between the number of clusters ($z$), scale coefficient ($\alpha$), code rate ($k/n$), and stripe width ($n$) in Figure \ref{fig:unilrc_codeinfo}. UniLRC easily achieves the target setting when $z\geq 10$, which corresponds to a moderate number of clusters in Google storage systems \cite{kadekodi2023practicaltos}.
For example, for $z=10, \alpha=2$, the rate of UniLRC ($210,180,20$) achieves $85.71\%$ with the stripe width $n=210$.
Therefore, UniLRC demonstrates a broad range of code rate feasibility based on practical system settings.

\textbf{Discussion.} For small-scale DSSs ($z\leq8$),  achieving a $0.85$ code rate with UniLRC's ``one local group, one cluster'' construction is challenging.
 However, we can relax this construction to a ``one local group, $t$ clusters''.
This method efficiently reduces the number of local parity blocks, thereby achieving a higher code rate. Although this method introduces cross-cluster traffic during degraded read and reconstruction, it only involves $t-1$ blocks of cross-cluster traffic, where $t=z/l$ and $l$ is the total number of local groups.

\textbf{Fault Tolerance Optimality.}
We now show that the minimum distance of UniLRC is exactly $r+2$.
\begin{theorem} \label{thm:code_distance}
    Given two positive integers $z$ (the number of clusters) and $\alpha$ (the scale coefficient), the minimum distance of UniLRC$(n = \alpha z^2+z, k = \alpha z^2-\alpha z, r=\alpha z)$ is exactly $r+2$. 
\end{theorem}

\textit{Sketch of Proof.}
Recall that from the code matrix in Definition \ref{def:code_matrix}, given the message $x$ of code $C$ and $n \times k$ generator matrix  $G$, we have $Gx = y$ where $y$ is the codeword.
In particular, there is also another code matrix, an $(n-k)\times n$ matrix  $H$, called a parity check matrix for $C$, which is defined as $$ C = \{x \in \mathbb{F}_q^n | Hx^T = 0\},$$
where $H$ can be derived from $G$.
From the generator matrix of \cref{equ:construction_unilrc}, we obtain the parity check matrix of UniLRC, given by 
$$H = \left[\begin{array}{c|c}
        A & I_{n-k}
    \end{array}\right], \text{ and }
    A = \left[\begin{array}{c}
        \mathcal{G}  \\
        \hline
        \mathcal{L}
    \end{array}\right]. $$
We aim to show that \textit{the parity check matrix $H$ has a set of $r+2$ linearly dependent columns but no set of $r + 1$ linearly dependent columns}.

First, we prove that any $r + 1$ columns are linearly independent.
To show this clearly, we split $H$ as three parts:  
$$ H =
\left [ \begin{array}{c|c|c}
    A & P_{G} & P_{L}
\end{array} \right], 
$$
where $P_{G}$ and $P_{L}$ corresponding to  global parity blocks and local parity blocks, respectively.
Given a matrix $T$ consisting of $r + 1$ columns from $H$, where $r + 1 = a+b+c$ means that we select $a,b$ and $c$ columns from $A, P_{G}$ and $P_{L}$, respectively.
Then we can simplify $T$ using elementary transformation in the following cases:
$$
T =
\begin{cases}
\left [ \begin{array}{c}
    V_a \\
    \hline
    O
\end{array} \right], & \text{if } a = r+1, b=c=0; \\
\left [ \begin{array}{c|c}
    V_a & O\\
    \hline
    O &I_{b}
\end{array} \right], & \text{if } a + b = r+1, b \neq 0, c=0; \\
\left [ \begin{array}{c|c|c}
    V_a & O & O\\
    \hline
    O & I_{b} & O \\
    \hline 
    O & O & I_{c}    
\end{array} \right], & \text{if } a + b + c = r+1, a \neq 0, b\neq 0, c\neq 0;
\end{cases}
$$
where $V$ is the Vandermonde matrix, $I$ is the identity matrix and $O$ is zero matrix, respectively.
From these cases, it is easy to verify $rank(T) = a + b + c= r + 1$, which proves the $r+1$ columns are linearly independent.

Second, it is easy to show that any set of $r+2$ columns is linearly dependent. For example, when $a = r+1, b=c=0$, we add a new information column to $T$ and simplify it through elementary transformations, obtaining:
$$T=\left [ \begin{array}{c}
    V \\
    \hline
    O
\end{array} \right]. $$
However, the $rank(T) = rank(V) = r + 1$, indicating that the $r + 2$ columns in $T$ are linearly dependent. \qed

Theorem \ref{thm:distance_optimal} shows that UniLRC is distance optimal, achieving the maximum possible fault tolerance.

\begin{theorem} \label{thm:distance_optimal}
     Given two positive integers $z$ (the number of clusters) and $\alpha$ (the scale coefficient), the UniLRC$(n = \alpha z^2+z, k = \alpha z^2-\alpha z, r=\alpha z)$ is distance optimal.
\end{theorem}

\begin{proof}
    UniLRC$(n = \alpha z^2+z, k = \alpha z^2-\alpha z, r=\alpha z)$ with $(r+1)|n$ and $r\geq d-2$, satisfies the condition of the distance optimal in Theorem \ref{thm:singleton}.
\end{proof}

\textbf{\textit{Recovery Locality} Optimality.}
We now demonstrate that UniLRC achieves optimal \textit{recovery locality}.

\begin{theorem} \label{thm:locality}
    The UniLRC$(n = \alpha z^2+z, k = \alpha z^2-\alpha z, r=\alpha z)$ achieves the minimum \textit{recovery locality}.
\end{theorem}

\begin{proof}
From Theorem~\ref{thm:code_distance}, the minimum distance of UniLRC is $d=r+2$. To tolerate a single cluster failure, we require \(d \geq \frac{n}{z} + 1\), i.e.,  $r=d-2 \geq \frac{n}{z} - 1$. Then we have 
\[
r \geq \frac{n}{z}-1 = \frac{\alpha z^2 + z}{z} - 1 = \alpha z  = r.
\]  
Thus, we have demonstrated the minimum locality parameter \(r\).
Consequently, UniLRC, with $\bar{r} = r$, also achieves the minimum \textit{recovery locality}. 
\end{proof}

\section{System Design of UniLRC} \label{sec:system_unilrc}
We design a UniLRC system to efficiently support common basic operations (\cref{subsec:basic_ops}), and implement a prototype to support UniLRC deployment (\cref{subsec:implementation}).

\begin{figure}[t!]
    \centering
    \includegraphics[scale=0.28]{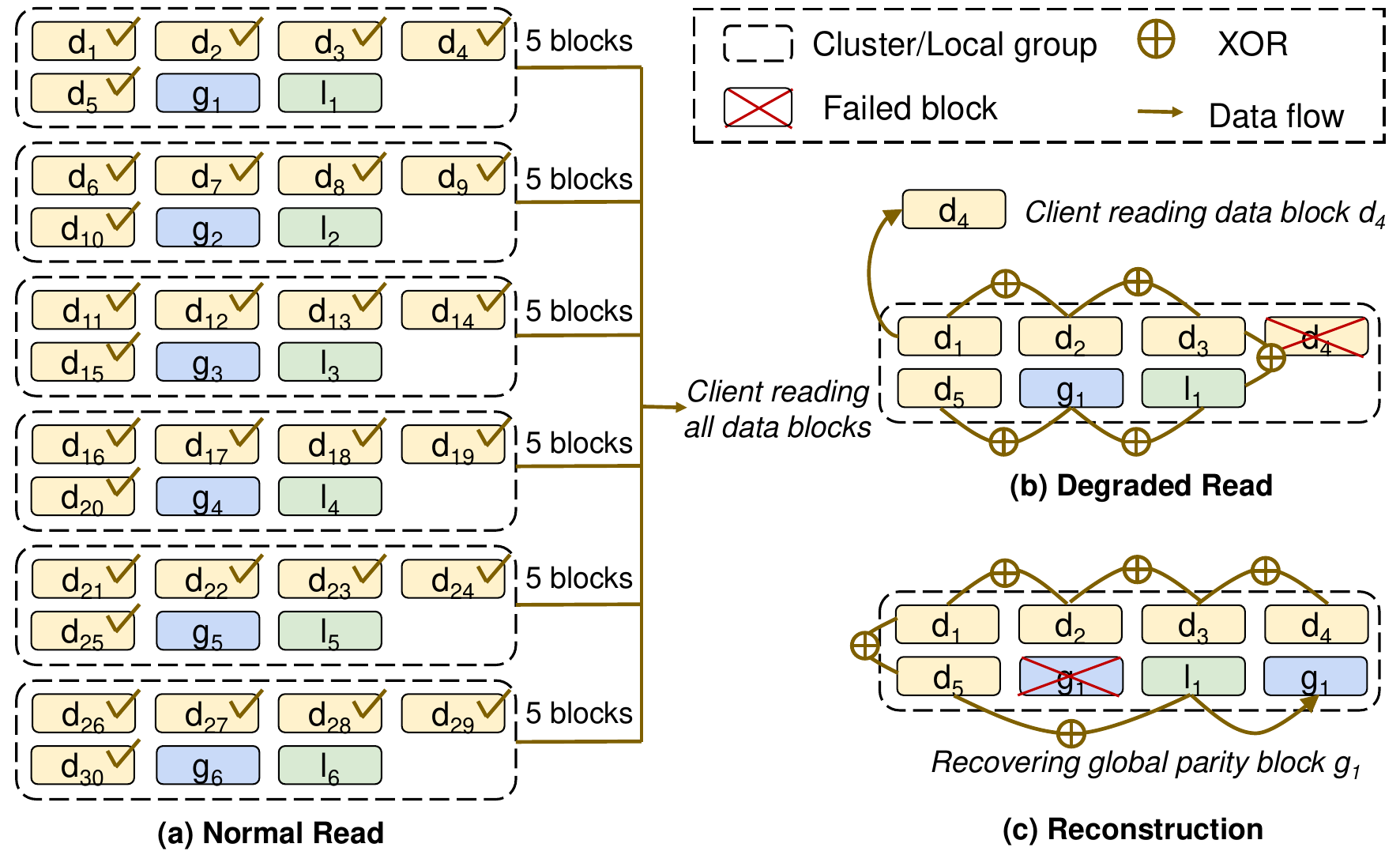}
    \caption{\textit{Normal read}, \textit{degraded read} and \textit{reconstruction} for UniLRC($42, 30, 6$).}
    \label{fig:unilrc_op_ex}
\end{figure}

\subsection{Common Basic Operations} \label{subsec:basic_ops}
We introduce the basic DSS operations of UniLRC, including \textit{normal read}, \textit{degraded read}, and \textit{reconstruction}, which are the most common operations in erasure-coded DSSs \cite{kadekodi2023practical,gan2025revisiting}.

\textbf{Normal read.} For \textit{normal read} operation, the client issues a read request for accessing all data blocks within one stripe. 

UniLRC distributes the $k$ data blocks evenly across all $z$ clusters, with each cluster holding $k/z$ data blocks.
Therefore, the UniLRC achieves optimal load balance across clusters, i.e., achieving the maximum read parallelism. 
Figure \ref{fig:unilrc_op_ex}(a) illustrates the \textit{normal read} operation for  UniLRC($42,30,6$), where all $30$ data blocks need to be read. Each cluster contributes to an equal number of blocks, i.e., $5$ blocks, resulting in a uniform distribution across $6$ clusters.
Below is the summarized property of \textit{normal read}.
\begin{tcolorbox}[colback=gray!5!white,colframe=gray!75!black,left=1mm, right=1mm, top=0.5mm, bottom=0.5mm, arc=1mm]
    \textbf{Property 1}: For \textit{normal read} operation, UniLRC achieves the maximum parallelism at the cluster level.
\end{tcolorbox}

The property guarantees that cross-cluster traffic for \textit{normal read} is evenly distributed across all clusters.

\textbf{Degraded read/Reconstruction.}
For \textit{degraded read} operation, the client request is issued to an unavailable data block, which is recovered from surviving blocks; and for \textit{reconstruction} (or single-block recovery), the DSS needs to recover a failed block, which could be either a data block or parity block.
For both \textit{degraded read} and \textit{reconstruction}, the main workflow is to recover a block by retrieving the surviving blocks within the same local group.

For recovering a failed block $b_f$ located in $group_j$, UniLRC retrieves the remaining blocks within the same local group and execute \texttt{XOR}, i.e.,  
$$b_f = \texttt{XOR}_{b_i\in group_j}b_i,$$ 
where $b_i$'s are the remaining blocks located in $group_j$, including data blocks, global parity blocks, and local parity blocks.
UniLRC distributes $z$ local groups across $z$ different clusters under the ``one local group, one cluster'' setting, so no cross-cluster traffic incurred, i.e., $\forall b_i\in group_j$.
Additionally, the \textit{recovery locality} $(\bar{r})$ of UniLRC is the minimum (as proven in Theorem \ref{thm:locality}).
Moreover, the encoding computation is performed via \texttt{XOR}, which avoids the high computational complexity of $GF$ multiplication.

Figure \ref{fig:unilrc_op_ex}(b) illustrates the \textit{degraded read} operation for  UniLRC ($42,30,6$), where the client performs a \textit{degraded read} of data block $d_4$. 
The decoding is  executed by XORing the surviving blocks $\{d_1, d_2, d_3, d_5, l_1, g_1\}$
and the \texttt{XOR}ed block is sent to the client.
This \textit{degraded read} operation incurs no cross-cluster traffic and only requires \texttt{XOR} calculations for $6$ blocks.
Figure \ref{fig:unilrc_op_ex}(c) shows the \textit{reconstruction} operation for  UniLRC($42,30,6$), where the the global parity block $g_1$ is to be recovered. Similar to the \textit{degraded read}, UniLRC \texttt{XOR}s the blocks $\{d_1, d_2, d_3, d_4, d_5, l_1\}$ within the same  local group, accessing only 6 blocks within the same cluster and resulting in no cross-cluster recovery traffic.
Below is the major property of \textit{degraded read} and \textit{reconstruction}.
\begin{tcolorbox}[colback=gray!5!white,colframe=gray!75!black,left=1mm, right=1mm, top=0.5mm, bottom=0.5mm, arc=1mm]
    \textbf{Property 2}: For \textit{degraded read} and \textit{reconstruction}, UniLRC achieves zero cross-cluster traffic, the minimum inner-cluster traffic and \texttt{XOR} decoding.
\end{tcolorbox}

The property guarantees (1) the minimum cross-cluster recovery traffic, i.e., zero cross-cluster traffic; (2) the minimum recovery traffic, i.e., accessing the minimum number of blocks, (3) the lowest decoding complexity, i.e., using only \texttt{XOR} decoding. In summary, combining \textit{Property 1} and \textit{Property 2}, we address the limitations (\#1, \#2, \#3) in \cref{sec:study_locality}.

\subsection{Implementation} \label{subsec:implementation}
\textbf{Prototype Architecture.} We develop a DSS prototype with UniLRC, as shown in Figure \ref{fig:implementation}.
The prototype consists of a coordinator, multiple clients, and multiple proxies, each managing several nodes. Clients upload data to the nodes and generate I/O requests, while the coordinator manages metadata (such as stripe-to-file and block-to-node mappings), stripe placement, and failure event handling. 
The proxy handles basic operations (such as \textit{normal read}, \textit{degraded read} and \textit{reconstruction}) and maintains the coding library. 
The coordinator and proxies monitor the cluster and server state information, respectively.
Our prototype is developed in C/C++ with around 12000 lines of code.

\begin{figure}[t!]
    \centering
    \includegraphics[scale=0.28]{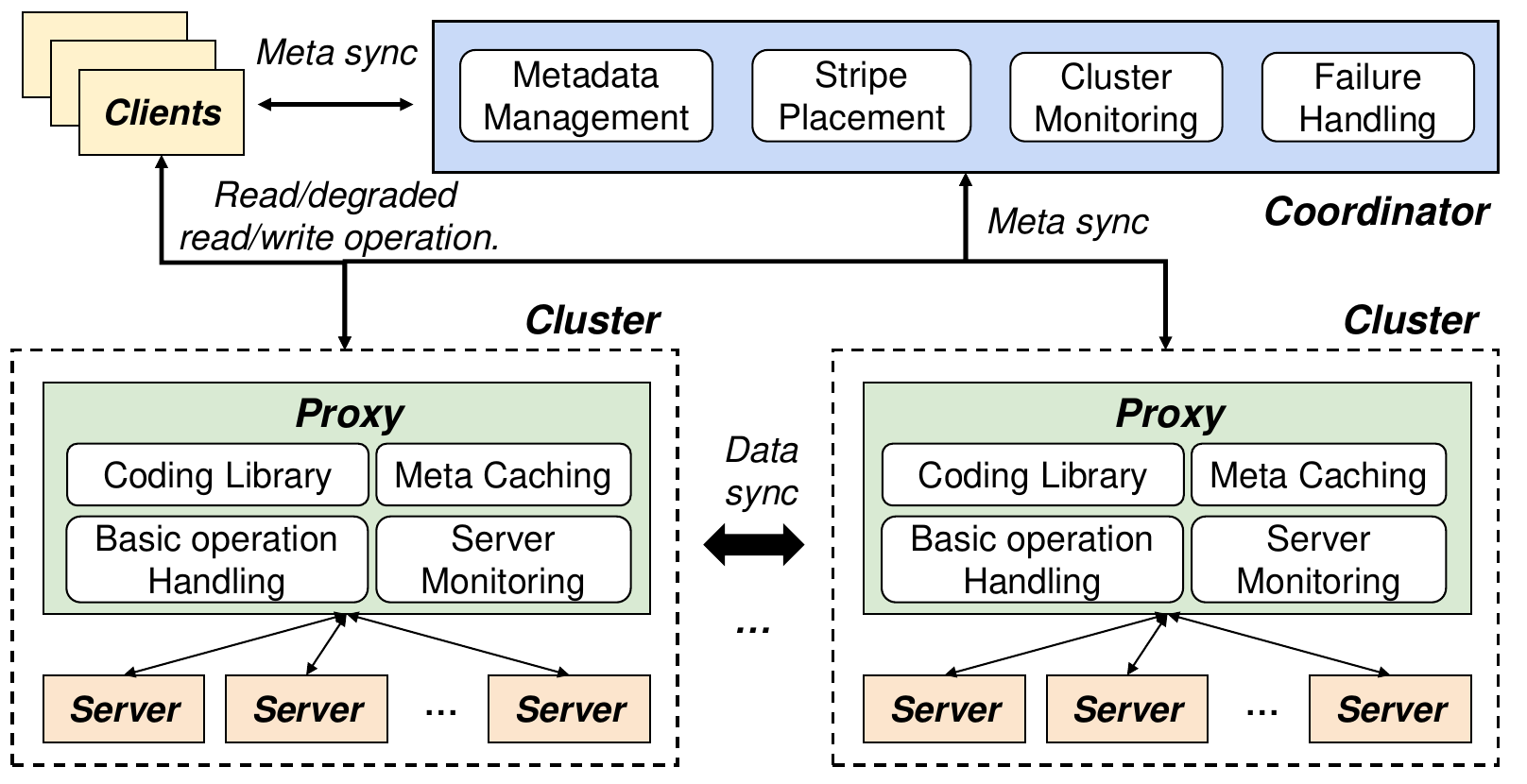}
    \caption{Prototype architecture of UniLRC.}
    \label{fig:implementation}
\end{figure}

\noindent\textbf{Coding Library.} We implement the UniLRC coding library using ISA-L \cite{isa_l}, a state-of-the-art acceleration library optimized with specialized instruction sets.
In our implementation, we employ the  \textit{GF($\mathit{2^8}$)} for coding computation with byte granularity.
UniLRC provides two simple coding APIs, \texttt{encode()} and \texttt{decode()}, and offers its coding functions as a shared library that can be dynamically linked into the address space of server processes, making it convenient for deployment in DSSs.

\noindent\textbf{Wide Stripe Deployment.} To support wide stripes  at multi-cluster deployment, we use a proxy node to simulate cluster management, with multiple processors corresponding to multiple servers. The prototype supports the number of clusters $z \leq 32$, and stripe width $n \leq 1024$, which aligning with the industry parameters \cite{kadekodi2023practicaltos,VAST}.

\section{Theoretical Analysis} \label{sec:the_ana}

\begin{table}[!t]
\centering
\caption{Code parameters for comparisons on wide LRCs.}
\begin{tabular}{p{1.5cm}<{\centering}|p{0.6cm}<{\centering}p{0.6cm}<{\centering}p{0.6cm}<{\centering}p{1.2cm}<{\centering}p{1.8cm}<{\centering}}
\hline
Scheme & $n$	&$k$	&$f$	&Rate	&UniLRC notes\\
\hline
30-of-42 & $42$	&$30$	&$7$ &$0.7143$ &$\alpha=1, z=6$	\\

112-of-136 & $136$	&$112$	&$17$ &$0.8235$	&$\alpha=2, z=8$\\

180-of-210 & $210$	&$180$	&$21$ &$0.8571$	&$\alpha=2, z=10$\\
\hline
\end{tabular}
\label{tab:wide_codes}
\end{table}

\textbf{Baseline codes and parameters.} We evaluate three representative LRCs that are widely deployed in practical systems, including Microsoft \cite{huang2012erasure} (ALRC) and Google \cite{kadekodi2023practical} (OLRC and ULRC). We demonstrate these codes with the $n=42, k=30$ parameters in Figure \ref{fig:baseline_lrc_example}.
We select three different stripe widths for wide LRCs, details of which are shown in Table \ref{tab:wide_codes}. For an apples-to-apples comparison, we fix the data size $k$, the code size $n$, and the fault tolerance requirement $f$, i.e., the ability to tolerate at least $f$ node failures and one cluster failure. The minimum distance $d$ of UniLRC, ULRC, and ALRC is equal to $d = f + 1$, while the OLRC cannot achieve $d = f + 1$ due to its construction limitations on large local groups.

\begin{figure*}
    \centering
    \includegraphics[scale=0.25]{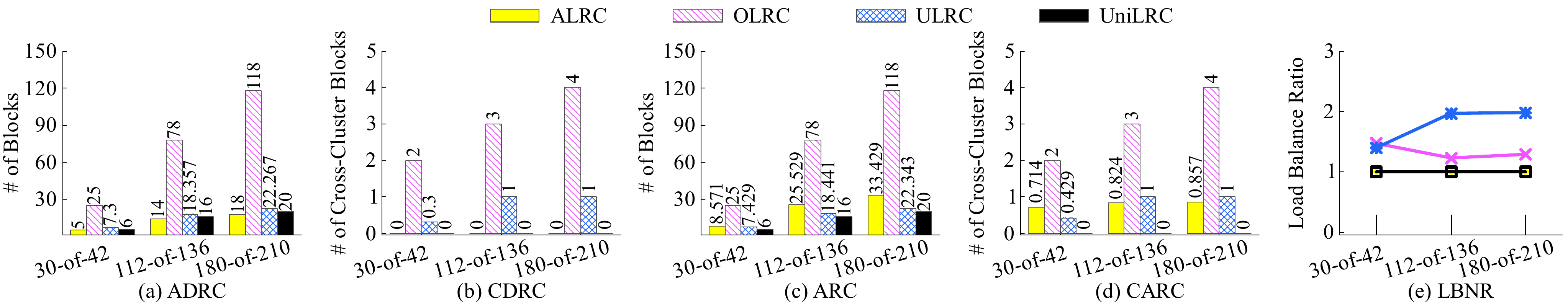}
    \caption{The performance on normal read, degraded read, reconstruction performance across all evaluated codes, measured by average degraded read cost (ADRC), cross-cluster average degraded read cost (CDRC), average recovery cost (ARC), cross-cluster average recovery cost (CARC) and load balance ratio of normal read (LBNR).}
    \label{fig:de_re_simulation}
\end{figure*}

\textbf{Metrics.}
To analyze the performance and reliability of UniLRC and baseline codes, we use six metrics in Table \ref{tab:metrics}.
To evaluate the recovery traffic, we define ADRC and ARC (\textit{recovery locality} $\bar{r}$ in \cref{subsubsection:code_locality}) for degraded read and reconstruction, respectively.
To evaluate cross-cluster recovery traffic, we define two cross-cluster metrics, CDRC and CARC.
For baseline codes, we adopt the replacement strategy used in ECWide \cite{hu2021exploiting}, a state-of-the-art placement strategy for wide LRCs specifically designed to optimize cross-cluster costs.
To evaluate the code reliability, we leverage mean-time-to-data-loss (MTTDL), a typical metric for reliability in the erasure-code system adopting the Markov model \cite{gibson1990redundant, huang2012erasure, hu2021exploiting, kadekodi2023practical}.
The entire set of results are shown in Figure \ref{fig:de_re_simulation} and Table \ref{tab:mttdl}, and we highlight the main points.

\begin{table}[!t]
\centering
\caption{The metrics of performance and reliability comparison for wide LRC. The $cost(b_i)$ denotes the number of blocks across nodes needed to reconstruct block $i$, and $cost^c(b_i)$ denotes the number of blocks across clusters required. The $cost^c$ denotes the number of access blocks for normal read.}
\begin{tabular}{p{6.2cm}<{\centering}|p{1.3cm}<{\centering}}
\hline
Metrics & Evaluation \\
\hline

Average degraded read cost (ADRC) & $\frac{\sum_{i=1}^{k}cost(b_i)}{k}$ \\
Cross-cluster average degraded read cost (CDRC) & $\frac{\sum_{i=1}^{k}cost^c(b_i)}{k}$ \\
Average recovery cost (ARC) & $\frac{\sum_{i=1}^{n}cost(b_i)}{n}$ \\
Cross-cluster average recovery cost (CARC) & $\frac{\sum_{i=1}^{n}cost^c(b_i)}{n}$ \\

Load balance ratio of normal read (LBNR) & $\frac{\max cost^{c}}{avg (cost^c)}$ \\

mean-time-to-data-loss (MTTDL) & Markov \\

\hline
\end{tabular}

\label{tab:metrics}
\end{table}

\textbf{ALRC has the lowest ADRC, closely followed by UniLRC.}
Since ALRC has the smallest number of blocks accessed for recovering data blocks, it also achieves the lowest ADRC. As the stripe width increases, the ADRC gap between ALRC and UniLRC narrows, owing to the wider stripe with larger local groups. For example, the ADRC of UniLRC is $20\%$ higher than that of ALRC for 30-of-42, but this difference narrows to $11\%$ for 180-of-210.


\textbf{UniLRC has the lowest CDRC, ARC and CARC.}
UniLRC achieves zero CDRC because each local group maps to a single cluster in its code construction.
Similarly, ALRC achieves zero CDRC by adopting the ECWide replacement strategy.
However, ECWide has limited optimization for CDRC in OLRC due to the large local groups
inherent in its code construction, causing 
a single local group distributed across multiple clusters.
UniLRC outperforms all other baseline codes across all metrics for reconstruction because it achieves zero cross-cluster traffic and the minimum \textit{recovery locality}.
Firstly, 
UniLRC achieves the optimal ARC by ensuring the minimum \textit{recovery locality} for all blocks. 
Secondly, UniLRC excels in CARC, thanks to its construction tailored to cluster topology, where each local group is mapped into a single cluster. This design eliminates cross-cluster traffic.

\textbf{Both ALRC and UniLRC have the optimal LBNR.}
ALRC and UniLRC achieve the optimal LBNR (equal to 1) by ensuring a uniform distribution of data blocks across clusters.
In contrast, although OLRC and ULRC use ECWide optimization to minimize CARC, their data placement strategies do not account for the normal read I/Os access pattern, leading to higher LBNR.

\begin{figure}[t!]
    \centering
    \includegraphics[scale=0.38]{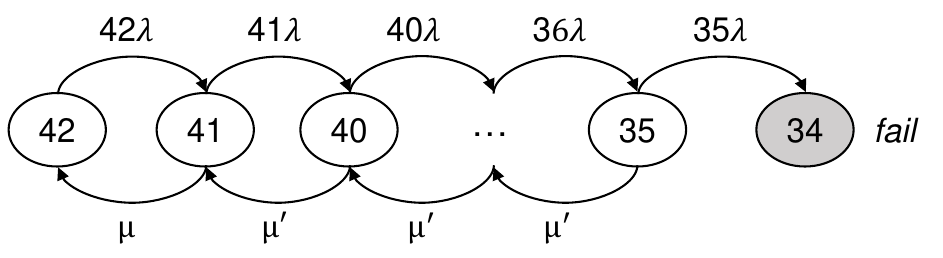}
    \caption{Morkov model for UniLRC$(42,30,6)$.}
    \label{fig:markov}
\end{figure}

\textbf{UniLRC strikes a strong balance between MTTDL and performance.}
Figure \ref{fig:markov} shows the Markov model for UniLRC$(42,30,6)$ code, with other baseline codes modeled in a similar fashion. Each state represents the number of available nodes in a stripe, where each block is stored on a node. For instance, State $42$ indicates that all blocks are available, whereas State $34$ signifies data loss. To simplify our analysis, we adopt the same assumptions as \cite{hu2021exploiting}, namely, independent node failures and the constraint of at most $f$ failed nodes.
Let $\lambda$ be the failure rate of each node. The transition rate from State $i$ to State $i-1$ (for $35 \leq i \leq 42$) is $i\lambda$, as any of the $i$ nodes in State $i$ can fail independently. For recovery, let $\mu$ be the recovery rate from State 41 to State 42, and $\mu'$ the recovery rate from State $i$ to State $i+1$ (for $36 \leq i \leq 40$).
Let $N$ be the total number of nodes, $S$ be the capacity of each node, $B$ be the network bandwidth per node, and $\epsilon$ be the fraction of bandwidth allocated for recovery. If a single node fails, the recovery load is shared by the remaining $N-1$ nodes, giving a total recovery bandwidth of $\epsilon(N-1)B$. Thus, $\mu = \frac{\epsilon(N-1)B}{CS}$, where $C$ is the recovery traffic per node.
For multiple node failures, $\mu' = 1/T$, where $T$ is the time to detect and trigger multi-node recovery, assuming that the recovery prioritized over single-node recovery \cite{huang2012erasure,hu2021exploiting}.
To compute the MTTDL, consider the example in Figure \ref{fig:markov}, where it is given by $\frac{42\lambda \bm{\cdot} \cdots \bm{\cdot} 35\lambda}{\mu \bm{\cdot} \cdots \bm{\cdot} \mu'}$.

We define $C = C_1 + \delta * C_2$, where $C_1$ is the cross-cluster traffic, $C_2$ is the inner-cluster traffic, and $\delta$ is the bandwidth coefficient (e.g., $\delta = 0.1$ means the cross-cluster bandwidth is one-tenth of the inner-cluster bandwidth).
For example, to recover a failed block in UniLRC$(42,30,6)$, the cross-cluster traffic is zero with $C_1 = 0$, and there are 6 uniform local groups with $C_2 = 6$. Thus, $C = 0 + \delta *C_2 = 0.6$ blocks with $\delta = 0.1$. This provides more accurate recovery traffic than previous methods \cite{hu2021exploiting}.
We set the default parameters as follows: $N = 400$, $S = 16$ TB, $\epsilon = 0.1$, $\delta = 0.1$, $T = 30$ minutes, $B = 1$ Gb/s, and $1/\lambda = 4$ years \cite{huang2012erasure,hu2021exploiting}.

Table \ref{tab:mttdl} compares the MTTDL of different LRC constructions. UniLRC shows an average MTTDL of $2.02\times$ and $1.71\times$ over ALRC and ULRC, respectively, as MTTDL is positively correlated with recovery traffic $C$. UniLRC minimizes cross-cluster recovery traffic ($C_1 = 0$) by eliminating cross-cluster traffic, and the small value of $r$  also reduces inner-cluster traffic ($C_2 = r$). Similarly, ULRC outperforms ALRC due to its lower $C$, as the global parity in ALRC requires all data blocks for recovery, thereby increasing both cross-cluster traffic $C_1$ and inner-cluster traffic $C_2$.
OLRC shows the highest MTTDL due to its larger minimum distance, which supports longer Markov chains. However, OLRC shows worse performance on normal read, degraded read and reconstruction, as shown in Figure \ref{fig:de_re_simulation}.
Therefore, UniLRC strikes a favorable balance between MTTDL and performance.

\begin{figure*}[htbp]
    \centering
    \includegraphics[scale=0.25]{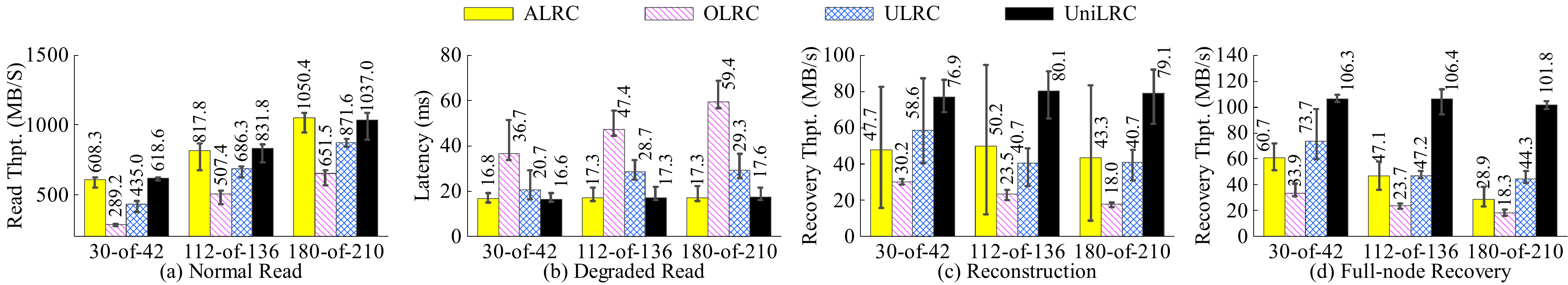}
    \caption{Basic operations performance on wide LRCs, including normal read, degraded read, and single-failure recovery.}
    \label{fig:sys_total_performance}
\end{figure*}

\section{System-Level evaluation} \label{sec:system_performance}
\textbf{Setup.} We conduct all experiments on 21 physical machines (1 client, 1 coordinator and 19 proxies) in the Wisconsin cluster of CloudLab \cite{duplyakin2019design}.
Each machine is equipped with a 6-core E5-2660v3 CPU, 160 GB DRAM, a 10Gb/s NIC, and a 1.2 TB SATA HDD. To enforce cross-cluster bandwidth constraints, we use the Linux traffic control tool Wondershaper \cite{wondershaper} to limit the outgoing bandwidth of each gateway. In our experiments, the cross-cluster bandwidth is set to 1Gb/s (1:10 bandwidth ratio), and the block size is configured as 1 MB (adopted in QFS \cite{ovsiannikov2013quantcast}).
We also use the same baseline LRCs from Microsoft \cite{huang2012erasure} (ALRC) and Google \cite{kadekodi2023practical}(OLRC and ULRC), with their parameters shown in Table \ref{tab:wide_codes}. 
For all baseline LRCs, we adopt the state-of-the-art ECWide replacement optimization \cite{hu2021exploiting}.
We report the average results over ten runs and provide the variance of these results.

\begin{table}[!h]
\centering
\caption{MTTDLs across all wide LRCs (years).}
\begin{tabular}{p{2cm}<{\centering}|p{1.2cm}<{\centering}p{1.2cm}<{\centering}p{1.2cm}<{\centering}p{1.2cm}<{\centering}}
\hline
\textbf{Scheme} & \textbf{ALRC} &  \textbf{OLRC} &\textbf{ULRC} & \textbf{UniLRC} \\
\hline

30-of-42  & 4.29e+10  & 3.24e+23 & 5.53e+10 & \textbf{9.62e+10}\\

112-of-136 & 3.47e+33   & 5.37e+49 & 4.12e+33 & \textbf{7.33e+33 } \\ 

180-of-210 & 1.64e+40   & 1.53e+60 & 1.74e+40 & \textbf{2.81e+40}\\
\hline

\hline
\end{tabular}

\label{tab:mttdl}
\end{table}

\textbf{Experiment 1 (Normal read).}
We evaluate normal read throughput under various k-of-n schemes in Table \ref{tab:wide_codes}. 
Figure \ref{fig:sys_total_performance}(a) demonstrates that UniLRC outperforms OLRC and ULRC in normal read performance. For example, UniLRC achieves an average increase of 27.46\% read throughput of ULRC.
The reason  is that UniLRC maximizes read parallelism (refer to Property 1 in \cref{subsec:basic_ops}).
Similarly, ALRC also shows good results with the ECWide \cite{hu2021exploiting} placement strategy.
Additionally, larger k-of-n schemes result in higher read throughput due to the increased data volume occupying the available bandwidth.
For example, with 180-of-210, the read throughput UniLRC and ALRC achieve the saturated network bandwidth of 1 Gbps cross-cluster bandwidth, while with 30-of-42, the read throughput UniLRC and ALRC achieve around 61.8\% of the available bandwidth.

\textbf{Experiment 2 (Degraded read).}
We evaluate the degraded read performance of a single unavailable data block in terms of the average degraded read latency under various k-of-n schemes.
Figure \ref{fig:sys_total_performance}(b) shows that UniLRC and ALRC outperform OLRC and ULRC in degraded read performance.
Both UniLRC and ALRC can avoid the cross-cluster traffic and minimize the ADRC (see Figure \ref{fig:de_re_simulation}(a)\&(b)) when recovering the data block, resulting in better degraded read latency.
The large local group of OLRC leads to a broad distribution across clusters, which results in its relatively poorer performance compared to other codes.
In summary, the UniLRC reduces the average degraded read latency by 33.15\% compared with ULRC, the state-of-the-art wide LRC.

\textbf{Experiment 3 (Single-failure recovery).}
We evaluate the recovery throughput for reconstruction (single-block recovery) under various k-of-n schemes.
Figure \ref{fig:sys_total_performance}(c) shows that UniLRC always outperforms the three baseline LRCs in recovery throughput, with an average increase of 67.18\%, 229.01\%, and 68.69\% over  ALRC, OLRC, and ULRC, respectively, for the same $n$ and $k$.
As $n,k$ increase, the recovery throughput of OLRC and ULRC decreases, while UniLRC remains stable.
This is because a larger k-of-n scheme introduces more cross-cluster traffic for OLRC and ULRC, while UniLRC still incurs zero cross-cluster traffic, which dominates the end-to-end recovery time.
Additionally, ALRC exhibits relatively large error bars due to its two types of reconstruction blocks, which incur different network traffic. Data and local parity blocks are recovered through local group reconstruction, while global parity blocks are recovered via global recovery.

For full-node recovery, we first generate the stripes and distribute them evenly across nodes, with around $40$ GB of data. We then turn off one node to trigger full-node recovery and measure the average recovery throughput.
Figure \ref{fig:sys_total_performance}(d) shows that all wide LRCs achieve a higher recovery throughput than in the reconstruction performance shown in Figure \ref{fig:sys_total_performance}(c).
This improvement is due to the higher parallelism in full-node recovery tasks compared with single-block recovery.
Additionally, all baseline LRCs and UniLRC exhibit lower recovery throughput as $(n,k)$ increases, due to larger overheads from network traffic and coding computation.
In summary, UniLRC always outperforms the three baseline LRCs in full-node recovery throughput, with an average increase of 130.04\%, 314.35\%, and 90.27\% over  ALRC, OLRC, and ULRC, respectively.

\begin{figure}[t!]
    \centering
    \includegraphics[scale=0.25]{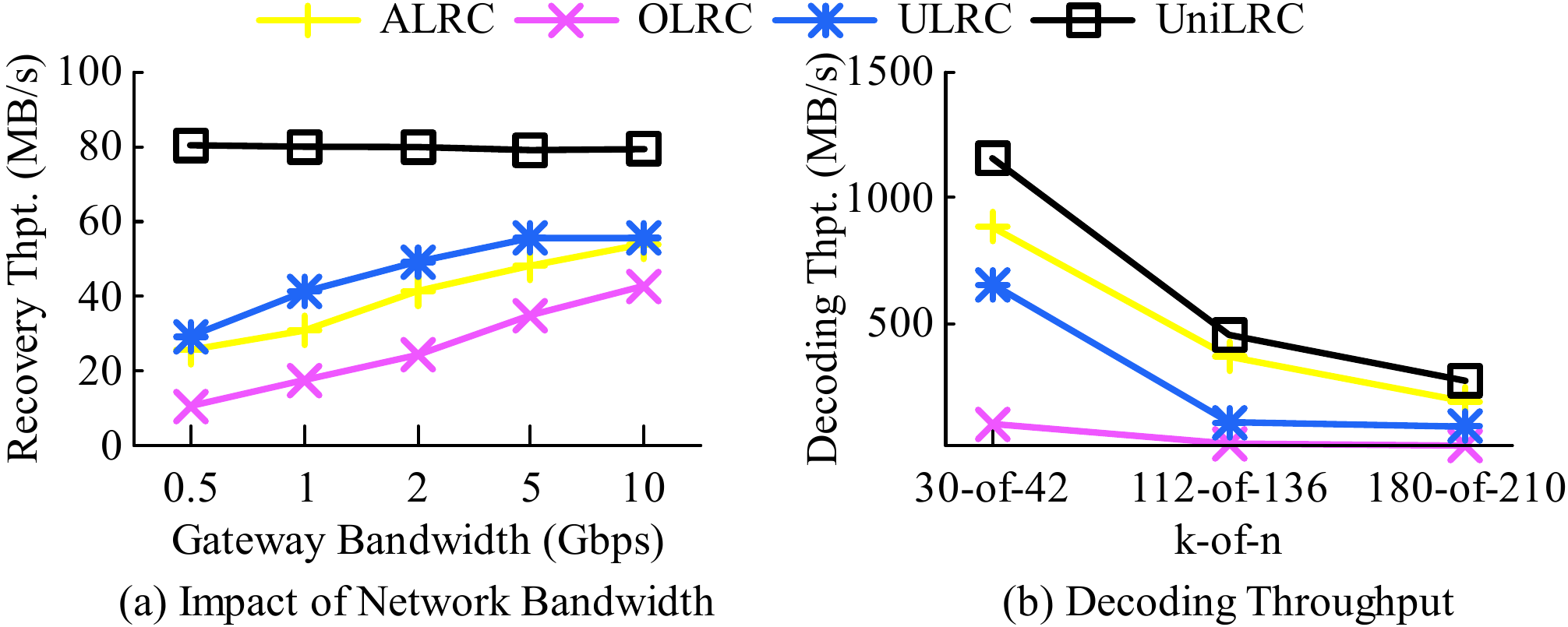}
    \caption{Experiment 4 \& Experiment 5: Comparison of reconstruction performance with varying cross-cluster bandwidth and decoding speed with different wide LRCs.}
    \label{fig:network_coding}
\end{figure}

\textbf{Experiment 4 (Impact of network bandwidth).}
We evaluate the reconstruction performance versus the cross-cluster network bandwidth under 180-of-210 scheme, varying the network bandwidth from 0.5 Gbps to 10 Gbps.
As shown in Figure \ref{fig:network_coding}(a), with the increase of cross-cluster network bandwidth, the recovery throughput of all baselines increases sharply, while that of UniLRC remains stable.
This is because UniLRC incurs no cross-cluster traffic during reconstruction.
Additionally, ULRC still exhibits a performance gap compared to UniLRC, even with abundant network bandwidth.
For example, with 10 Gbps cross-cluster bandwidth (unlimited), UniLRC achieves a $42.66\%$ increase in recovery throughput compared with ULRC.
The main reason is that UniLRC has lower recovery cost due to its the minimum \textit{recovery locality} Theorem \ref{thm:locality}.



\textbf{Experiment 5 (Decoding performance).}
We evaluate the decoding performance in terms of decoding throughput under various k-of-n schemes. Figure \ref{fig:network_coding}(b) shows that UniLRC outperforms all baseline codes, with an average decoding throughput of $1.33 \times$, $19.03 \times$, and $3.05 \times$  over  ALRC, OLRC, and ULRC, respectively.
The OLRC shows the lowest decoding throughput due to its large local groups, which results in a large volume of data being retrieved. 
The ALRC performs worse than UniLRC due to the higher cost and complexity of decoding the global parity block.

\begin{figure}[t!]
    \centering
    \includegraphics[scale=0.25]{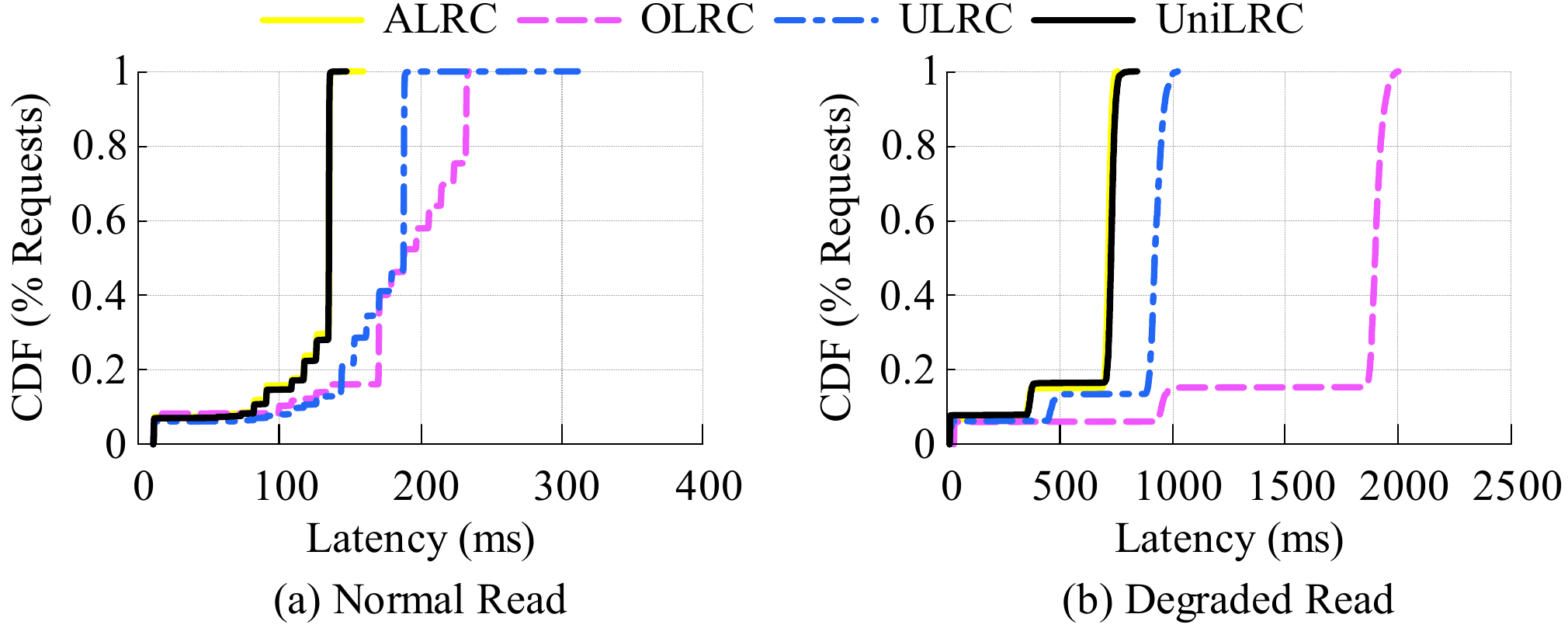}
    \caption{Production workload performance on normal read and degraded read under 180-of-210 scheme.}
    \label{fig:workload}
\end{figure}
\textbf{Experiment 6 (Production workload).}
We evaluate the normal read and degraded read performance using a production object store workload with variable sizes \cite{vajha2018clay,rashmi2016ec}. Following the configuration in \cite{vajha2018clay}, we select objects of medium (1MB), medium/large (32MB), and large (64MB) sizes, distributed in respective proportions of 82.5\%, 10\%, and 7.5\%, in line with the Facebook data analytics cluster reported in \cite{rashmi2016ec}. The block size and code parameters are set to 1MB and $n=210,k=180$, respectively, ensuring that each stripe consists of multiple objects.
We first generate the data stripes with the specified object sizes and ratios, then place the stripes in a round-robin manner. Finally, the client issues normal read and degraded read requests.

Figure \ref{fig:workload} shows the cumulative distribution function (CDF) of requests for normal read latency and degraded read latency over 1000 requests. We observe that UniLRC and ALRC outperform ULRC and OLRC in both normal and degraded read scenarios. Specifically, compared to ULRC, the state-of-the-art wide LRC, UniLRC reduces the average normal read latency by 25.89\% and the average degraded read latency by 23.23\%, respectively.
The excellent read performance of UniLRC can be attributed to its optimal network and computation complexity properties (refer to Properties 1 and 2 in \cref{subsec:basic_ops}).

\section{Conclusion}\label{sec:conclusion}
Existing wide LRCs struggle to achieve optimal fault tolerance while maintaining high performance.
In this paper, we introduce UniLRC, a novel family of wide LRCs designed to optimize both performance and reliability. UniLRCs are constructed using a generator matrix, followed by matrix decompositions that tightly couple local and global parity blocks. This approach ensures distance optimality while addressing the locality limitations.
Compared to state-of-the-art wide LRCs, UniLRC strikes a better balance between reliability and performance.

\bibliographystyle{ACM-Reference-Format}
\bibliography{bibfile}


\end{document}